\documentclass[a4paper,USenglish,cleveref, autoref, thm-restate]{lipics-v2021}

\usepackage{tabularx}
\usepackage{booktabs}

\usepackage{mathtools}
\usepackage{wasysym}
\usepackage{nicefrac}

\usepackage{xcolor}
\usepackage{todonotes}

\usepackage[sort&compress,numbers]{natbib}

\definecolor{mdarkred}{RGB}{165,0,38}
\definecolor{mred}{RGB}{215,25,28}
\definecolor{mdarkorange}{RGB}{244,109,67}
\definecolor{morange}{RGB}{253,174,97}
\definecolor{myellow}{RGB}{254,224,144}
\definecolor{mverylightblue}{RGB}{224,243,248}
\definecolor{mlightblue}{RGB}{171,217,233}
\definecolor{mblue}{RGB}{116,173,209}
\definecolor{mdarkblue}{RGB}{69,117,180}
\definecolor{mverydarkblue}{RGB}{49,54,149}

\usepackage{tikz}
\usetikzlibrary{decorations, decorations.pathreplacing, decorations.pathmorphing, patterns, calc, arrows, arrows.meta}

\usepackage{xspace}

\usepackage{csquotes}

\newcommand{\problem}[1]{\textnormal{#1}\xspace}

\newcommand{\rar}[1]{\textnormal{RAR}$(#1)$}
\newcommand{\lrs}[1]{\textnormal{LRS}$(#1)$}

\newcommand{\rai}{\problem{RAI}}

\newcommand{\tailoredsat}{\problem{$3$-SAT$^*$}}

\newcommand{\jobs}{\mathcal{J}}
\newcommand{\machs}{\mathcal{M}}
\newcommand{\emachs}{\machs}

\newcommand{\ress}{\mathcal{R}}
\newcommand{\Cmax}{C_{\max}}
\newcommand{\Cmin}{C_{\min}}

\newcommand{\Opt}{\operatorname{\text{\textsc{opt}}}}

\newcommand{\bigBorder}{\xi}

\newcommand{\largejobs}{\mathcal{L}}

\newcommand{\hugejobs}{\mathcal{H}}
\newcommand{\smalljobs}{\mathcal{S}}

\newcommand{\bordermachs}{\mathcal{B}}
\newcommand{\hugemachs}{\mathcal{X}}
\newcommand{\candidatemachines}{\mathcal{C}}

\newcommand{\true}{\top}
\newcommand{\false}{\bot}

\newcommand{\Oh}{\mathcal{O}}

\newcommand{\clausemach}{\mathtt{CMach}}
\newcommand{\truthmach}{\mathtt{TMach}}
\newcommand{\sortmach}{\mathtt{SMach}}

\newcommand{\clausejob}{\mathtt{CJob}}
\newcommand{\truthjob}{\mathtt{TJob}}
\newcommand{\variablejob}{\mathtt{VJob}}

\newcommand{\bridgejob}{\mathtt{BJob}}

\newcommand{\gatejob}{\mathtt{GJob}}
\newcommand{\sortjob}{\mathtt{SJob}}

\newcommand{\ampshiftjob}{\mathtt{ASJob}}
\newcommand{\ampbridgejob}{\mathtt{ABJob}}

\newcommand{\fgatemach}{\mathtt{FGMach}}
\newcommand{\bgatemach}{\mathtt{BGMach}}
\newcommand{\fsortmach}{\mathtt{FSMach}}
\newcommand{\bsortmach}{\mathtt{BSMach}}
\newcommand{\ampmach}{\mathtt{AMach}}

\newcommand{\tblock}{\mathcal{T}}

\newcommand{\sblock}{\mathcal{S}}
\newcommand{\cblock}{\mathcal{C}}

\newcommand{\eps}{\varepsilon}

\DeclarePairedDelimiter\floor{\lfloor}{\rfloor}
\DeclarePairedDelimiter\ceil{\lceil}{\rceil}

\DeclarePairedDelimiter\set{\lbrace}{\rbrace}
\DeclarePairedDelimiterX\sett[2]{\lbrace}{\rbrace}{ #1 \,\delimsize| \,\mathopen{} #2 }

\title{(In-)Approximability Results for Interval, Resource Restricted, and Low Rank Scheduling} 

\titlerunning{Interval, Resource Restricted, and Low Rank Scheduling} 

\author{Marten Maack}{Heinz Nixdorf Institute \& Department of Computer Science, Paderborn University, Paderborn, Germany}{marten.maack@hni.uni-paderborn.de}{https://orcid.org/0000-0001-7918-6642}{}

\author{Simon Pukrop}{Heinz Nixdorf Institute \& Department of Computer Science, Paderborn University, Paderborn, Germany}{simonjp@mail.uni-paderborn.de}{https://orcid.org/0000-0002-4473-5215}{}

\author{Anna Rodriguez Rasmussen}{Department of Mathematics, Uppsala University, Uppsala, Sweden}{anna.rodriguez-rasmussen@math.uu.se}{}{}

\authorrunning{M. Maack and S. Pukrop and A. Rodriguez Rasmussen} 

\Copyright{Marten Maack and Simon Pukrop and Anna Rodriguez Rasmussen}

\ccsdesc[100]{Theory of computation~Scheduling algorithms}
\ccsdesc[100]{Theory of computation~Problems, reductions and completeness} 

\keywords{Scheduling, Restricted Assignment, Approximation, Inapproximability} 

\category{} 

\relatedversion{} 

\funding{This work was supported by the German Research Foundation (DFG) within the Collaborative Research Centre “On-The-Fly Computing” under the project number 160364472 — SFB 901/3}

\nolinenumbers 

\hideLIPIcs  

\EventEditors{John Q. Open and Joan R. Access}
\EventNoEds{2}
\EventLongTitle{42nd Conference on Very Important Topics (CVIT 2016)}
\EventShortTitle{CVIT 2016}
\EventAcronym{CVIT}
\EventYear{2016}
\EventDate{December 24--27, 2016}
\EventLocation{Little Whinging, United Kingdom}
\EventLogo{}
\SeriesVolume{42}
\ArticleNo{23}

\begin{document}

\maketitle

\begin{abstract}
We consider variants of the restricted assignment problem where a set of jobs has to be assigned to a set of machines, for each job a size and a set of eligible machines is given, and the jobs may only be assigned to eligible machines with the goal of makespan minimization.
For the variant with interval restrictions, where the machines can be arranged on a path such that each job is eligible on a subpath, we present the first better than $2$-approximation and an improved inapproximability result.
In particular, we give a $(2-\frac{1}{24})$-approximation and show that no better than $9/8$-approximation is possible, unless P=NP. 
Furthermore, we consider restricted assignment with $R$ resource restrictions and rank $D$ unrelated scheduling.
In the former problem, a machine may process a job if it can meet its resource requirements regarding $R$ (renewable) resources.
In the latter, the size of a job is dependent on the machine it is assigned to and the corresponding processing time matrix has rank at most $D$.
The problem with interval restrictions includes the 1 resource variant, is encompassed by the 2 resource variant, and regarding approximation the $R$ resource variant is essentially a special case of the rank $R+1$ problem.
We show that no better than $3/2$, $8/7$, and $3/2$-approximation is possible (unless P=NP) for the 3 resource, 2 resource, and rank 3 variant, respectively.
Both the approximation result for the interval case and the inapproximability result for the rank 3 variant are solutions to open challenges stated in previous works.
Lastly, we also consider the reverse objective, that is, maximizing the minimal load any machine receives, and achieve similar results. 
\end{abstract}

\section{Introduction}\label{sec:introduction}

Makespan minimization on unrelated parallel machines, or unrelated scheduling for short, is considered a fundamental problem in approximation and scheduling theory.
In this problem, a set $\jobs$ of jobs has to be assigned to a set $\machs$ of machines via a schedule $\sigma: \jobs \rightarrow \machs$.
Each job $j$ has a processing time $p_{ij}$ depending on the machine $i$ it is assigned to and the goal is to minimize the makespan $\Cmax(\sigma) = \max_{i\in\machs} \sum_{j\in\sigma^{-1}(i)}p_{ij}$. 
In 1990, Lenstra, Shmoys, and Tardos \cite{DBLP:journals/mp/LenstraST90} presented a 2-approximation for this problem and further showed that no better than $1.5$-approximation can be achieved (unless P$=$NP) already for the restricted assignment problem, where each job $j$ has a size $p_j$ and $p_{ij} \in \set{p_j,\infty}$ for each machine $i$.
For each job $j$ we denote its set of eligible machines by $\emachs(j) = \sett{i\in\machs}{p_{ij} = p_j}$.
Closing or narrowing the gap between 2-approximation and $1.5$-inapproximability is a famous open problem in approximation \cite{DBLP:books/daglib/0030297} and scheduling theory \cite{SchuurmanW99}.
The present paper deals with certain subproblems of both unrelated scheduling and restricted assignment.

\subparagraph{Interval Restrictions.}

In the variant of restricted assignment with interval restrictions, denoted as \rai in the following, there is a total order of the machines and each job $j$ is eligible on a discrete interval of machines, i.e., $\machs = \set{M_1,M_2,\dots,M_m}$ and $\emachs(j) = \set{M_{\ell},M_{\ell+1},\dots M_r}$ for some $\ell,r\in[m]$.
There are several variants and special cases of this problem that are known to admit a polynomial time approximation scheme (PTAS), see \cite{DBLP:journals/eor/LiW10a,DBLP:journals/orl/MuratoreSW10,EpsteinL11,DBLP:journals/tcs/JansenMS20,DBLP:phd/dnb/Schwarz10,DBLP:journals/corr/KhodamoradiKRS16}, the most prominent of which is probably the hierarchical case \cite{DBLP:journals/eor/LiW10a} in which each job is eligible on an interval of the form $\set{M_{1},M_{2},\dots M_r}$, i.e., the first machine is eligible for each job. 
For \rai, on the other hand, there is an $(1+\delta)$-inapproximability result for some small but constant $\delta>0$ \cite{DBLP:conf/stacs/MaackJ20}.
Furthermore, Schwarz \cite{DBLP:phd/dnb/Schwarz10} designed a $(2-2/(\max_{j\in\jobs}p_j))$-approximation (assuming integral processing times); and Wang and Sitters \cite{DBLP:journals/ipl/WangS16} studied an LP formulation that provides an optimal solution for the special case with two distinct processing times and some additional assumption.

\subparagraph{Resource Restrictions.}

In the restricted assignment problem with $R$ resource restrictions, or \rar{R}, a set $\ress$ or $R$ (renewable) resources is given, each machine $i$ has a resource capacity $c_{r}(i)$ and each job $j$ has a resource demand $d_{r}(j)$ for each $r\in\ress$.
The eligible machines are determined by the corresponding resource constraints, i.e., $\emachs(j) =\sett[\big]{i\in\machs}{\forall r\in\ress:d_{r}(j) \leq c_{r}(i)}$ for each job $j$.
It is easy to see, that \rar{1} corresponds to the mentioned hierarchical case which admits a PTAS \cite{DBLP:journals/eor/LiW10a}.
On the other hand, there can be no approximation algorithm with ratios smaller than $48/47 \approx 1.02$ or $1.5$ for \rar{2} and \rar{4}, respectively, unless P=NP, see \cite{DBLP:conf/stacs/MaackJ20}.
The same paper also points out that the case with one resource is a special case of the interval case which in turn is a special case of the two resource case, i.e., \rar{1} $\subset$ \rai $\subset$ \rar{2}.
While the hierarchical case, i.e. \rar{1}, has been studied extensively before, \rar{R} was first introduced in a work by Bhaskara et al.~\cite{DBLP:conf/soda/BhaskaraKTW13}, who mentioned it as a special case of the next problem that we consider.

\subparagraph{Low Rank Scheduling.}

In the rank $D$ version of unrelated scheduling, or \lrs{D}, the processing time matrix $(p_{ij})$ has a rank of at most $D$.
Alternatively (see \cite{DBLP:conf/stacs/0011MYZ17}), we can assume that each job $j$ has a $D$ dimensional size vector $s(j)$ and each machine $i$ a $D$ dimensional speed vector $v(i)$ such that $p_{ij} = \sum_{k=1}^{D}s_k(j)v_k(i)$.
Now, \lrs{1} is exactly makespan minimization on uniformly related parallel machines, which is well known to admit a PTAS \cite{DBLP:journals/jacm/HochbaumS87}.
Bhaskara et al.~\cite{DBLP:conf/soda/BhaskaraKTW13}, who introduced \lrs{D}, presented a QPTAS for \lrs{2} along with some initial inapproximability results for $D>2$.
Subsequently, Chen et al.~\cite{DBLP:journals/orl/ChenYZ14} showed that there can be no better than $1.5$-approximation for \lrs{4} unless P=NP, and for \lrs{3} the same authors together with Marx \cite{DBLP:conf/stacs/0011MYZ17} ruled out a PTAS.
On an intuitive level, resource restrictions can be seen as a restricted assignment version of low rank scheduling.
However, there is a more direct relationship between the two problems: 
for each \rar{R} instance there exist \lrs{R+1} instances that are arbitrarily good approximations of the former (see \cite{DBLP:conf/stacs/MaackJ20}).
Hence, any approximation algorithm for \lrs{R+1} can also be used for \rar{R}, and any inapproximability result for \rar{R} carries over to \lrs{R+1}.
In fact many (but not all) inapproximability results for low rank scheduling essentially have this form.

\subparagraph{Results.}

We present improved approximation and inapproximability results for this family of problems.
In particular:
\begin{itemize}
\item An approximation algorithm for \rai with ratio $2-\frac{1}{24}\approx 1.96$ presented in \cref{sec:algorithms};
\item a reduction that rules out a better than $1.5$ approximation unless P=NP, i.e., a $1.5$-inapproximability result, for \rar{3} presented in \cref{sec:rar3_reduction}; 
\item a $8/7$-inapproximability result for \rar{2} presented in \cref{sec:rar2_reduction};
\item a $9/8$-inapproximability result for \rai presented in \cref{sec:rai_reduction};
\item and a $1.5$-inapproximability result for \lrs{3} presented in \cref{sec:rank3}.
\end{itemize}
The positive result for \rai can be considered the first of the two main contributions of this paper.
Finding a better than 2-approximation for \rai was posed as an open challenge in previous works \cite{DBLP:journals/ipl/WangS16,DBLP:journals/tcs/JansenMS20,DBLP:phd/dnb/Schwarz10}.
When considering the respective results in \cite{DBLP:journals/ipl/WangS16} and \cite{DBLP:phd/dnb/Schwarz10}, in particular, it seems highly probable that the actual goal of the research was to address exactly that challenge. 
The presented approximation algorithm follows the approach of solving and rounding a relaxed linear programming formulation of the problem, which has been used in the classical work by Lenstra et al. \cite{DBLP:journals/mp/LenstraST90} and many of the results thereafter.
In particular, we extend the so called assignment LP due to Lenstra et al. \cite{DBLP:journals/mp/LenstraST90} and design a customized rounding approach.
Both the linear programming extension and the rounding approach utilize extensions and refinements of ideas from \cite{DBLP:phd/dnb/Schwarz10} and \cite{DBLP:journals/ipl/WangS16}.
Our result joins the relatively short list of special cases of the restricted assignment problem that do not allow a PTAS and for which an approximation algorithm with rate smaller than 2 is known. 
Other notable entries are the restricted assignment problem with only two processing times \cite{DBLP:conf/soda/ChakrabartyKL15} and the so-called graph balancing case \cite{DBLP:journals/algorithmica/EbenlendrKS14}, where each job is eligible on at most two machines.

The inapproximability results directly build upon the results presented in the paper \cite{DBLP:conf/stacs/MaackJ20}, which in turn utilizes many of the previously published ideas, e.g., from \cite{DBLP:journals/algorithmica/EbenlendrKS14,DBLP:conf/soda/BhaskaraKTW13,DBLP:journals/orl/ChenYZ14,DBLP:conf/stacs/0011MYZ17,DBLP:journals/mp/LenstraST90}.
We use the satisfiability problem presented in \cite{DBLP:conf/stacs/MaackJ20} as the starting point for all of our reductions.
For the \rai result in particular, we refine and restructure the respective results from \cite{DBLP:conf/stacs/MaackJ20} aiming for a significantly better ratio.
The respective reduction involves a sorting process and curiously the main improvement in the reduction involves changing a sorting process resembling insertion sort into one resembling bubble sort.
Due to this change, the construction becomes locally less complex enabling the use of smaller processing times and hence a stronger inapproximability result.
Furthermore, the simplified construction in the result enables us to use the basic structure of the reduction as a starting point for the second main result of the paper, namely, the $1.5$-inapproximability result for \rar{3}.
For this reduction several additional considerations and gadgets are needed, arguably making it the most elaborate of the presented results.
The search for an inapproximability result with a reasonably big ratio for \rar{3} was stated as an open challenge in the long version of \cite{DBLP:conf/stacs/0011MYZ17}.
Adding the new result yields a very clear picture regarding the approximability of low rank makespan minimization: 
There is a PTAS for \lrs{1}, a QPTAS for \lrs{2}, and a 1.5-inapproximability result for \lrs{D} with $D\geq 3$. 
The last two reductions regarding \rar{2} and \rar{3} yield much improved inapproximability results for the respective problems using comparatively simple and elegant reductions.
The result regarding \rar{3}, in particular, closes a gap in the results of \cite{DBLP:conf/stacs/MaackJ20} and also yields an (arguably) easier, alternative proof for the result of \cite{DBLP:journals/orl/ChenYZ14}.
Finally, we note that all of the inapproximability results regarding restricted assignment with resource restrictions can be directly applied to the so called fair allocation or santa claus versions of the problems.
In these problem variants, we maximize the minimum load received by the machines rather than minimization of the maximum load, i.e., the objective function is given by $\Cmin(\sigma) = \min_{i\in\machs} \sum_{j\in\sigma^{-1}(i)}p_{ij}$ in this case.

\subparagraph{Related Work.}

We refer to \cite{DBLP:conf/stacs/MaackJ20}, the corresponding long version \cite{DBLP:journals/corr/abs-1907-03526}, and the references therein for a more detailed discussion of related work and only briefly discuss some further references.
Regarding the problem of closing the gap between the $2$-approximation and $1.5$-approximability, there has been a promising line of research (e.g. \cite{DBLP:journals/siamcomp/Annamalai19,DBLP:journals/siamcomp/JansenR20}) in the last decade, starting with a breakthrough result due to Svensson \cite{DBLP:journals/siamcomp/Svensson12}, which in turn was preceded by corresponding results for the fair allocation (santa claus) version of the problem, see, e.g.,~\cite{DBLP:conf/stoc/BansalS06,DBLP:conf/soda/Feige08}.
These results are based on local search algorithms that usually do not run in polynomial time, but can be used to prove a small integrality gap for a certain linear program, which therefore can be used to approximate the optimum objective value in polynomial time without actually producing a schedule.
In the online setting, versions of restricted assignment with different types of restrictions, and variants of \rar{1} in particular, have been intensively studied.
We refer to the surveys \cite{LeungL08survey,LeungL16update,DBLP:journals/anor/LeeLP13} for an overview.
Lastly, we note that the low rank scheduling has also been considered from the perspective of fixed-parameter tractable algorithms \cite{DBLP:conf/stacs/0011MYZ17}.

\section{Approximation Algorithm for Makespan Minimization with Interval Restrictions}\label{sec:algorithms}

In this section, we establish the first approximation for \rai with an approximation factor better than $2$:
\begin{theorem}
	\label{the:approx2minusgamma}
	There is a $(2-\gamma)$-approximation for \rai with $\gamma = \frac{1}{24}$.
\end{theorem}
The particular value of the parameter $\gamma$ is justified in the end.
To achieve this result, we first formulate a customized linear program based on the assignment LP due to Lenstra et al. \cite{DBLP:journals/mp/LenstraST90} and develop a rounding approach that places different types of jobs in phases.
Note that the placement of big jobs with size close to $\Opt$ (where $\Opt$ is the makespan of an optimal schedule) is often critical when aiming for an approximation ratio of smaller than $2$ for a makespan minimization problem. 
For instance, the classical 2-approximation \cite{DBLP:journals/mp/LenstraST90} for restricted assignment produces a schedule of length at most $\Opt + \max_{j\in\jobs} p_j$ where $\Opt$ is the makespan of an optimal schedule and hence the approximation ratio is better if $\max_{j\in\jobs} p_j$ is strictly smaller than $\Opt$.
This is also the case with our approach -- the main effort goes into the careful placement of such big jobs.
In particular, we place the largest jobs in a first rounding step and the remaining big jobs in a second.
All of these jobs have the property that each machine should receive at most one of them and they are placed accordingly. 
Moreover, the placement is designed to deviate not too much from the fractional placement due to the LP solution.
In a last step, the remaining jobs are placed.
Each rounding step is based on a simple heuristic approach that considers the machines from left to right and places the least flexible eligible jobs first, i.e., the jobs that have not been placed yet, are eligible on the current machine, and have a minimal last eligible machine in the ordering of the machines. 
Both the LP and the rounding approach reuse ideas from \cite{DBLP:phd/dnb/Schwarz10,DBLP:journals/ipl/WangS16}.
Hence, the main novelty lies in the much more elaborate approach for placing the mentioned big jobs in two phases. 

In the following, we first establish some preliminary considerations; then briefly discuss the least flexible first heuristic utilized in the rounding approach; next we formulate the LP and argue that it is indeed a relaxation of the problem at hand; and then discuss and analyze the different phases of the rounding procedure step by step.

\subparagraph{Preliminaries.}

For any integer $k$, we set $[k] = \set{0,\dots,k-1}$.
We apply the standard technique (see \cite{DBLP:journals/mp/LenstraST90}) of using a binary search framework to guess a candidate makespan $T$.
The goal is then to either correctly decide that no schedule with makespan $T$ exists, or to produce a schedule with makespan at most $(2-\gamma)T$.
Given this guess $T$, we divide the jobs $j$ into small ($p_j \leq 0.5T$), large ($0.5T < p_j \leq (0.5+\xi)T$) and huge ($(0.5+\xi)T<p_j$) jobs depending on some parameter $\xi = \frac{1}{24}$ which is justified later on.
We denote the sets of small, large, and huge jobs as $\smalljobs$, $\largejobs$, and $\hugejobs$, respectively.
Furthermore, we fix the (total) order of the machines such that each job is eligible on consecutive machines.
This is possible since we are considering \rai.
For the sake of simplicity, we assume $\machs = [m]$ with the ordering corresponding to the natural one and set $\machs(\ell,r) = \set{\ell,\dots, r}$ for each $\ell,r\in\machs$.
When considering the machines, we use a left to right intuition with predecessor machines on the left and successor machines on the right. 
Note, that for each job $j$ there exists a left-most and right-most eligible machine and we denote these by $\ell(j)$ and $r(j)$, respectively, i.e., $\emachs(j) = \machs(\ell(j),r(j))$.
For a set of jobs $J\subseteq\jobs$, we call a job $j\in J$ \emph{least flexible} in $J$ if $r(j)$ is minimal in $\sett{r(j')}{j'\in J}$, and a job $j$ is called \emph{less flexible} than a job $j'$ if $r(j)\leq r(j')$.
Lastly, we set $J(\ell,r) = \sett[\big]{j\in J}{\emachs(j) \subseteq \machs(\ell,r)}$ for each set of jobs $J\subseteq\jobs$ and pair of machines $\ell,r\in\machs$, and $p(J) = \sum_{j\in\jobs} p_j$.

\subparagraph{Least Flexible First.}

Consider the least flexible first heuristic for \rai:
The optimum makespan $\Opt$ is lower bounded by the maximum job size $\max_{j\in\jobs}p_j$ as well as the average load $p(\jobs(\ell,r))/|\machs(\ell,r)|$ of jobs that have to be placed in any given interval of machines $\machs(\ell,r)$ in a feasible schedule.
Let $L\leq\Opt$ be the maximum of all of the above lower bounds.  
Starting with the left-most machine in the ordering, the heuristic works as follows:
\begin{itemize}
\item Let $i^*$ be the current machine and $J$ the set of jobs that have not been placed yet and are eligible on $i^*$.
\item If $i^*$ has received a load of at most $L$ up to now and $J\neq \emptyset$, place a \emph{least flexible} job $j\in J$ on $i^*$, i.e., a job $j\in J$ with minimal $r(j)$, and consider $i^*$ again.
\item Otherwise, consider the next machine in the ordering or stop if there is none.
\end{itemize}
It is easy to see that this simple approach yields a 2-approximation:
\begin{lemma}
The least flexible first heuristic places each job and each machine receives a load of at most $L + \max_{j\in\jobs}p_j \leq 2\Opt$.
\end{lemma}
\begin{proof}
The heuristic obviously never places a load greater than $L + \max_{j\in\jobs}p_j$ on any machine.
Now assume for the sake of contradiction that there exists a job that is not placed by this heuristic.
Let $j^*$ be a job that is not placed, i.e., after considering the right-most eligible machine $r^* = r(j^*)$ the job $j^*$ has not been placed by the algorithm.
Then $r^*$ did receive a load greater than $L$ and we have $r(j)\leq r^*$ for each job $j$ placed on $r^*$.
Let $\ell^*\leq r^*$ be the left-most machine with the properties that each machine in $\machs(\ell^*,r^*)$ did receive a load greater than $L$ and $r(j)\leq r^*$ for each job $j$ placed on $\machs(\ell^*,r^*)$.
Furthermore, let $J^*$ be the set of jobs placed on $\machs(\ell^*,r^*)$ by the algorithm together with $j^*$.
Then $\ell(j) \geq \ell^*$ for each $j\in J^*$ since otherwise there exists a machine $i$ directly preceding $\ell^*$ that could have received $j$ as well but did not.
This would imply that $i$ did receive a load greater than $L$ of jobs less flexible than $j$ yielding a contradiction to the choice of $\ell^*$. 
Hence, $J^*\subseteq \jobs(\ell^*,r^*)$ yielding the contradictory statement $p(\jobs(\ell^*,r^*)) \geq p(J^*) > |\machs(\ell^*,r^*)|L \geq p(\jobs(\ell^*,r^*))$ (considering the definition of $L$).
\end{proof}
We are not aware of this observation being published before, but consider it very likely that it was already known, in particular, since variants thereof are used in \cite{DBLP:journals/ipl/WangS16,DBLP:phd/dnb/Schwarz10}.

\subparagraph{Linear Program.}

The classical assignment LP (see \cite{DBLP:journals/mp/LenstraST90}) is given by assignment variables $x_{ij}\in [0,1]$ for each $i\in \machs$ and $j \in \jobs$ and the following constraints:
\begin{align}
\sum_{i\in \machs}x_{ij} &=1 & \forall j\in\jobs \label{lp:jobs} \\
\sum_{j\in \jobs} p_jx_{ij} &\leq T & \forall i\in\machs \label{lp:machs} \\
x_{ij} &=0 & \forall j\in\jobs, i\in\machs \setminus \machs(j) \label{lp:eligible} 
\end{align}
\cref{lp:jobs} guarantees that each job is (fractionally) placed exactly once;
\cref{lp:machs} ensures that each machine receives at most a load of $T$; 
and due to \cref{lp:eligible} jobs are only placed on eligible machines.
We add additional constraints that have to be satisfied by any integral solution.
In particular, we add the following constraints using parameters $UB(\ell,r)$ for each $\ell,r\in\machs$ with $\ell\leq r$, which will be properly introduced shortly:
\begin{align}
\sum_{j\in \largejobs\cup\hugejobs} x_{ij} &\leq 1 & \forall i\in\machs \label{lp:largeorhugepermach}\\
\sum_{i\in \machs(\ell,r)}\sum_{j \in \hugejobs}x_{ij} &\leq UB(\ell,r) & \forall \ell,r\in\machs, \ell\leq r \label{lp:hugeinterval} 
\end{align}
\cref{lp:largeorhugepermach} captures the simple fact that no machine may receive more than one job of size larger than $0.5T$ and was used in \cite{DBLP:journals/algorithmica/EbenlendrKS14} as well.
The bound $UB(\ell,r)$, on the other hand, is defined in relation to the total load of small jobs that has to be scheduled in the respective interval $\machs(\ell, r)$.
In particular, we consider the overall load of small jobs that have to be placed in the interval together with the load due to huge jobs with their sizes rounded down to their minimum size.
The respective load has to be bounded by $T$ times the number of machines in the interval, i.e., $\sum_{i\in\machs(l,r)}\sum_{ j\in \hugejobs}(0.5+\bigBorder)Tx_{ij} + p(\smalljobs(\ell,r)) \leq T|\machs(l,r)|$.
Since the number of huge jobs placed in an interval is integral for an integral solution, we can therefore set $UB(\ell,r) = \floor[\big]{(T|\machs(l,r)| - p(\smalljobs(\ell,r))) / ((0.5+\bigBorder)T) }$.
We note that a constraint similar to \cref{lp:hugeinterval} is also used in \cite{DBLP:journals/ipl/WangS16,DBLP:phd/dnb/Schwarz10}.
Summing up, we try to solve the linear program given by \cref{lp:jobs,lp:machs,lp:eligible,lp:hugeinterval,lp:largeorhugepermach} which is indeed a relaxation for~\rai.
If this is not successful, we reject $T$ and otherwise round the solution $x$ using the procedure described in the following and yielding a rounded solution $\bar{x}$.

\subparagraph{Placement of Huge Jobs.}

Starting with the first machine in the ordering, we place the huge jobs as follows:
\begin{itemize}
\item Let $i^*$ be the current machine and $H$ the set of huge jobs that have not been placed yet and are eligible on $i^*$.
\item If $\floor[\big]{\sum_{i\in\machs(0,i^*)} \sum_{j\in\hugejobs} x_{ij}} > \floor[\big]{\sum_{i\in\machs(0,i^*-1)} \sum_{j\in\hugejobs} x_{ij}}$ and $H\neq \emptyset$, place a least flexible job $j\in H$ on $i^*$, i.e., we set $\bar{x}_{i^*j} = 1$.
\item Consider the next machine in the ordering or stop if there is none.
\end{itemize}
We denote the set of machines that are considered by the above procedure as $\hugemachs $, i.e., $\hugemachs = \sett[\big]{i^*\in\machs}{\floor[\big]{\sum_{i\in\machs(0,i^*)} \sum_{j\in\hugejobs} x_{ij}}  > \floor[\big]{\sum_{i\in\machs(0,i^*-1)} \sum_{j\in\hugejobs} x_{ij}}}$.
This procedure indeed works and we preserve a connection to the original LP solution:
\begin{lemma}\label{lem:approx:hugeBound}
All of the huge jobs are placed (on eligible machines) by the above procedure and, for each $\ell,r\in\machs$ with $\ell\leq r$, we have $\sum_{i\in \machs(\ell,r)}\sum_{j \in \hugejobs}\bar{x}_{ij} \leq \ceil[\big]{ \sum_{i\in \machs(\ell,r)}\sum_{j \in \hugejobs}x_{ij} }$.
\end{lemma}
\begin{proof}
The second statement directly follows from the fact that we only place a new job if the sum of fractional huge jobs placed up to the current machine in the LP solution reaches a new integer.
Regarding the first, assume for the sake of contradiction that there exists a huge job $j^*$ that is not placed.
We set $r^* = r(j^*)$.
Note that $\hugemachs\cap\emachs(j^*)\neq\emptyset$ since $j^*$ was placed fractionally by the LP and for the same reason the last such machine $r'\leq r^*$ or some predecessor did receive some of the fractional load of $j^*$ in the LP as well. 
Then $r'$ did receive a huge job $j$ with $r(j)\leq r(j^*)$.
Let $\ell^*$ be the left-most machine such that each machine in $\machs(\ell^*,r^*)\cap \hugemachs$ did receive a huge job $j$ with $r(j) \leq r(j^*)$ and let $H^*$ be the set of huge jobs placed by the procedure on $\machs(\ell^*,r^*)\cap \hugemachs$ together with $j^*$.
Then we have $\ell^*\leq \ell(j)$ for each $j\in H^*$ since otherwise there exist a machine $i\in\hugemachs$ directly preceding $\ell^*$ that may have received a job from $H^*$.
Since this did not happen it must have received a less flexible job, which is a contradiction to the choice of $\ell^*$.
Hence, $H^* \subseteq \hugejobs(\ell^*,r^*)$ but $|\machs(\ell^*,r^*)\cap \hugemachs| < |H^*|$.
This is a contradiction since each job in $H^*$ was completely placed in $\machs(\ell^*,r^*)$ by the LP which implies $|\machs(\ell^*,r^*)\cap \hugemachs| \geq |H^*|$.  
\end{proof}
We note that this first rounding step is very similar to the first rounding step in \cite{DBLP:journals/ipl/WangS16}.

\subparagraph{Mapping out the Regions.}

In the next step, we divide the machines into regions, where each region did receive fractional large load of (roughly) one.
To that end, we define a set of border machines $\bordermachs$ as the machines considered from left to right where the sum of fractionally placed large jobs hits a new integer, i.e.,
$\bordermachs = \sett[\big]{i'\in\machs}{\floor[\big]{\sum_{i\in\machs(0,i')} \sum_{j\in\largejobs} x_{ij}} > \floor[\big]{\sum_{i\in\machs(0,i'-1)} \sum_{j\in\largejobs} x_{ij}}}$.
Moreover, let $\bordermachs = \set{i_1,\dots,i_q}$ with $i_1<\dots<i_q$ and $i_0$ the left-most machine with $\sum_{j\in\largejobs}x_{i_0j} > 0$. 
For each $s\in [q] =\set{0,\dots, q-1}$, we may initially define the $s$-th region as $R^s = \machs(i_s, i_{s+1})$.
At this point consecutive regions overlap by one machine.
We want to change this, while guaranteeing that each region retains at least one \emph{candidate machine} that may receive a large job in the following.
In particular, a machine $i\in\machs$ is a candidate if it did receive some fractional large or huge job in the LP solution, i.e., $\sum_{j\in\hugejobs\cup\largejobs}x_{ij} > 0 $, but no huge job afterwards, i.e., $\sum_{j\in\hugejobs}\bar{x}_{ij} = 0 $.
We denote the set of candidate machines as $\candidatemachines$.
For each $s\in[q-1]$, we apply the following procedure in incremental order:
\begin{itemize}
\item Check whether region $R^s$ needs the last machine to have at least one candidate, i.e., $\machs(i_s,i_{s+1} -1)\cap \candidatemachines = \emptyset$.
\item If this is the case, we set $R^{s+1} = \machs(i_{s+1} +1,i_{s+2})$ and otherwise set $R^{s} = \machs(i_{s} ,i_{s+1} -1)$. 
\end{itemize}
After applying this procedure, we have:
\begin{lemma}
	\label{lem:approx:1MachPerRegion}
	The regions are non-overlapping and each contain at least one candidate.
\end{lemma}
\begin{proof}
The first statement is obvious and we show the second via contradiction.
Assume that there is a region $R^{r}$ without a candidate machine.
After running the procedure the original left and right borders $i_{s}$ and $i_{s+1}$ of a region $R^s$ may or may not be included in $R^s$ and we set $\mathtt{inner}(R^s) = \machs(i_{s} + 1, i_{s+1} -1)$ for each $s\in[q]$.
Since $R^{r}$ does not contain candidates, we know that $\mathtt{inner}(R^{r})$ cannot contain candidates, $i_{r +1}$ was assigned to $R^{r}$ by the algorithm, and is not a candidate either.
Let $\ell\leq r$ be maximal with the property that $R^{\ell}$ did receive $i_{\ell}$ in the algorithm and $R^s$ did receive $i_{s+1}$ for each $s\in\set{\ell,\dots,r}$. 
Then the rules of the algorithm imply that $i_{\ell}$ is not a candidate and $\mathtt{inner}(R^{s})$ does not contain a candidate either for each $s\in\set{\ell,\dots,r}$. 
Hence, the only possible remaining candidate machines in the respective regions are the borders $\bordermachs \cap \machs(i_{\ell +1}, i_r)$.
Let $C$ be the set of candidates in the respective regions, i.e., $C= \candidatemachines\cap \bigcup_{s\in\set{\ell,\dots,r}}R_s$, and $k = |\set{\ell,\dots,r}|$. 
Then the above implies $|C| \leq k -1$.
For the remainder of the proof, we introduce some additional notation: the set of assigned huge jobs in $\bigcup_{s\in\set{\ell,\dots,r}}R_s$ is given by $H$ and the fractional number of large or huge jobs placed in these regions according to $x$ is denoted as $\mathtt{fracLarge}$ or $\mathtt{fracHuge}$, respectively, i.e., $\mathtt{fracLarge} = \sum_{i\in\machs(i_{\ell},i_{r+1})} \sum_{j\in\largejobs} x_{ij}$ and $\mathtt{fracHuge} = \sum_{i\in\machs(i_{\ell},i_{r+1})} \sum_{j\in\hugejobs} x_{ij}$.
Since $i_{\ell} \in R^{\ell}$ and $i_{r+1}\in R^r$, the definition of the borders yields $\mathtt{fracLarge} \geq k$.
Note that each machine in the regions that did receive a fractional large or huge job in the LP solution but is not a candidate subsequently received a huge job.
Hence, we have
\[|H| + |C| \overset{(\ref{lp:largeorhugepermach})}{\geq} \mathtt{fracHuge} + \mathtt{fracLarge} \geq \mathtt{fracHuge} + k\]
and therefore $|H| \geq \mathtt{fracHuge} + 1$.
However, \Cref{lem:approx:hugeBound} gives us $|H| \leq \lceil \mathtt{hugeLoad} \rceil < \mathtt{hugeLoad} + 1$. \lightning
\end{proof}
Before proceeding with the placement of the large jobs, we note the following technical observation:
\begin{lemma}\label{lem:intervals_regions_frac_load}
Let $\ell,r\in\machs$ with $\ell\leq r$, $\ell\in R^s$, $r\in R^t$, $k = |\set{s,\dots,t}|$, and $\mathtt{fracLarge} = \sum_{i\in \machs(\ell,r)}\sum_{j \in \largejobs}x_{ij}$. 
Then we have $ k-2 < \mathtt{fracLarge} < k+2$.
Furthermore, $\mathtt{fracLarge} < k+1$ if either $\ell >i_s$ or $r < i_{t+1}$ and $\mathtt{fracLarge} < k$ if both of these conditions hold.
\end{lemma}
\begin{proof}
There are at least $k-2$ regions that are completely included in $\machs(\ell,r)$ including their original outer borders. 
Hence, the definition of the regions yields $k-2 < \mathtt{fracLarge} $.
For the remaining statements, we consider the definition of the regions more closely. 
Note that there exist numbers $\lambda_u\in[0,1)$ and $\rho_u \in (0,1]$ for each $u\in[q]$ such that $\lambda_u + \big(\sum_{i\in\mathtt{inner}(R^u)}\sum_{j\in\largejobs} x_{ij}\big) +\rho_u = 1$ (using the notation of the last proof); 
and furthermore $\sum_{j\in\largejobs} x_{i_0j} = \lambda_{0}$ if $i_0\neq i_1$, $\sum_{j\in\largejobs} x_{i_uj} = \rho_{u-1} + \lambda_{u}$ for $u \in\set{1,\dots,q-1}$, and $\sum_{j\in\largejobs} x_{i_qj} = \rho_{q-1}$.
We assume for now $s>0$ and $t<q-1$.
Then we have $\mathtt{fracLarge} \leq \rho_{s-1} + k + \lambda_{t+1} < k+2$, 
$\mathtt{fracLarge} \leq k + \lambda_{t+1} < k+1$ if $\ell >i_s$, 
$\mathtt{fracLarge} < \rho_{s-1} + k \leq k+1$ if $r < i_{t+1}$, and $\mathtt{fracLarge} <  k$ if both of the conditions hold.
If $s=0$ or $t=q$, we can prove the statement analogously.
\end{proof}

\subparagraph{Placement of Large Jobs.}

Using the regions, we place the large jobs via the following procedure starting with the first region:
\begin{itemize}
\item Let $R^*$ be the current region and $L$ the set of large jobs that have not been placed yet and are eligible on at least one candidate machine from $R^*$.
\item Do the following \emph{twice}: Pick a least flexible large job $j\in L$, place it on the leftmost eligible candidate machine $i\in R^*$, i.e. $\bar{x}_{ij} = 1$, and update $L$.
\item Consider the next region in the ordering or stop if there is none.
\end{itemize}
Observe that the placement of both the large and huge jobs guarantees that only machines that did receive fractional large or huge load in the LP solution may receive any large or huge job and each such machine receives at most one such job. 
We argue that this procedure works and also retains some connection to the original LP solution $x$.
\begin{lemma}\label{lem:approx:largeBound}
All large jobs are placed (on eligible machines) by the described procedure and, for each $\ell,r\in\machs$ with $\ell\leq r$, we have $\sum_{i\in \machs(\ell,r)}\sum_{j \in \largejobs}\bar{x}_{ij} < 2(\sum_{i\in \machs(\ell,r)}\sum_{j \in \largejobs}x_{ij}  + 2)$.
\end{lemma}
\begin{proof}
Regarding the second statement note that we place at most 2 jobs in each region and hence \cref{lem:intervals_regions_frac_load} directly yields the proof.
As usual, we proof the first statement by contradiction.
To that end, assume that there exists a large job $j^*$ that is not placed by the procedure.
First note, that there is at least one eligible candidate machine for $j^*$.
To see this, consider the set $M$ of eligible machines $i\in\emachs(j)$ that either received fractional load of $j^*$ or some huge load, i.e., $x_{ij}>0$ for $j\in\set{j^*}\cup \hugejobs$.
Then \cref{lp:largeorhugepermach} implies $\sum_{i\in M} \sum_{j\in \hugejobs} x_{ij} \leq |M| - 1$.
Hence, at most $|M| - 1$ many huge jobs are placed on machines from $M$ due to \Cref{lem:approx:hugeBound} and therefore at least one of these machines is a candidate. 
There are two possibilities why $j^*$ was not placed on such a machine: either a less flexible job got placed on the machine, or two other less flexible jobs were already placed in the same region.
Let $r^* = r(j^*)$ and $\ell^*\leq \ell$ be minimal with the property that each large job that was placed in $\machs(\ell^*,r^*)$ is less flexible than $j^*$ and each free candidate machine in the interval is free because two other machines in the same respective region already received a large job less flexible than $j^*$.
Furthermore, let $J^*$ be the set of large jobs placed in $\machs(\ell^*,r^*)$ together with $j^*$.
We argue that $\ell(j)\geq \ell^*$ for each $j\in J^*$.
Otherwise, there exists a job $j\in J^*$ eligible on machine $\ell^* - 1$.
Then there are three possibilities regarding this machine.
It was not a candidate before the procedure; it was a candidate received a job less flexible then $j$ (and therefore also less flexible then $j^*$); or it was a candidate and did not receive a large job because two other machines in the same region received a job less flexible then $j$.
Each yields a contradiction to the definition of $\ell^*$.
Let $\mathtt{fracLarge} = \sum_{i\in\machs(\ell^*,r^*)} \sum_{j\in\largejobs} x_{ij}$ be the sum of fractional large jobs in $\machs(\ell^*,r^*)$ according to $x$.
Note that we did show $J^* \subseteq \jobs(\ell^*,r^*)$ and hence $\mathtt{fracLarge}\geq |J^*|$.

Let $M^*\subseteq\machs(\ell^*,r^*)$ be the set of machines that did receive a fraction of a job from $J^*\cup\hugejobs$.
Then \cref{lp:largeorhugepermach} implies $\sum_{i\in M^*} \sum_{j\in \hugejobs} x_{ij} \leq |M^*| - |J^*|$, and furthermore \Cref{lem:approx:hugeBound} yields that at most $|M^*| - |J^*|$ huge jobs are placed on machines from $|M^*|$.
Hence, there are at least $|J^*|$ candidate machines in $\machs(\ell^*,r^*)$.
Since not all of the jobs from $J^*$ have been placed by the procedure, there is therefore at least one free machine in $i^*\in\machs(\ell^*,r^*)$.
The definition of $\ell^*$ yields, that two jobs less flexible then $j^*$ have been placed in the same region as $i^*$ and these jobs have to be included in $J^*$ (and the machines they are placed on in $\machs(\ell^*,r^*)$).

We now take a closer look at the regions (partially) included in $\machs(\ell^*,r^*)$.
Let $\ell^* \in R^s$, $r^*\in R^t$, and $k = |\set{s,\dots,t}|$.
We consider three cases:
If we have $\ell^* = i_{s}$ and $r^* = i_{t+1}$, i.e., the borders of the interval correspond to the (original) outer borders of their regions, then each of the regions $R^s,\dots, R^t$ did receive at least one job from $J^*$ and one received at least two yielding $k \leq |J^*| - 2 \leq \mathtt{fracLarge} - 2$.
Moreover, if $\ell^* > i_{s}$ or $r^* < i_{t+1}$, then one of the regions $R^s,\dots, R^t$ may not have received a job from $J^*$ changing the inequality to $k \leq \mathtt{fracLarge} - 1$.
Lastly, if both $\ell^* > i_{s}$ and $r^* < i_{t+1}$, then the two outer regions may have received no job from $J^*$ yielding $k \leq \mathtt{fracLarge}$.
However, \cref{lem:intervals_regions_frac_load} considers the same three cases, yielding $\mathtt{fracLarge}< k + 2$, $\mathtt{fracLarge}< k + 1$, and $\mathtt{fracLarge}< k$, respectively. \lightning
\end{proof}

\subparagraph{Placement of Small Jobs.}

Lastly we place the small jobs.
Starting with the first machine, we do the following:
\begin{itemize}
\item Let $i^*$ be the current machine and $J$ the set of jobs that have not been placed yet and are eligible on $i^*$.
\item Successively place least flexible jobs $j$ on $i^*$, i.e., set $\bar{x}_{i^*j} = 1$, until either $J = \emptyset$ or placing the next job would raise the load of $i^*$ above $(2-\gamma)T$. 
\item Consider the next machine in the ordering or stop if there is none.
\end{itemize}
We argue that this procedure works under certain conditions:
\begin{lemma}
All small jobs are placed (on eligible machines) by the described procedure if $\gamma \leq \xi$, $\gamma + \xi \leq \frac{1}{12}$, and $8\xi + 7\gamma \leq 0.75$ hold.
In the resulting schedule, each machine has a load of at most $(2-\gamma)T$.
\end{lemma}
\begin{proof}
For the sake of easier presentation, we assume $T=1$ in the following (this can be established via scaling).
As usual, the second statement is easy to see and we prove the first via contradiction.
To that end, let $j^*$ be a small job we cannot place.
Let $\mathtt{load}(i) = \sum_{j\in\jobs} p_j\bar{x}_{ij}$ be the load machine $i\in\machs$ did receive.
We call a machine \emph{full} if we stop placing small jobs on it because placing another job would have caused a load of more than $2-\gamma$. 
Note that $\mathtt{load}(i) > 1.5 - \gamma$ for full machines $i\in\machs$.
Let $r^* = r(j^*)$.
Then $r^*$ is full since we were not able to place $j^*$ and we have $\mathtt{load}(i) + p_{j^*} > 2-\gamma$.
Moreover, all the small jobs placed on $r^*$ are less flexible than $j^*$.
Let $\ell^*$ be the left-most machine with the property that each machine in $\machs(\ell^*,r^*)$ is full and each small job placed on such a machine is less flexible then $j^*$, and let $S^*$ be the set of small jobs placed on $\machs(\ell^*,r^*)$ together with $j^*$.
We have $\ell(j)\geq \ell^*$ for each $j\in S^*$ since otherwise machine $\ell^* -1$ has to be full and each small job placed on this machine must be less flexible then $j$ yielding a contradiction to the choice of~$\ell^*$.
Hence, we have $S^* \subseteq \smalljobs(\ell^*,r^*)$.

We establish some further notation. 
Let $\mathtt{fracHuge}$, $\mathtt{fracLarge}$, and $\mathtt{fracSmall}$, be the summed up number of fractional huge, large, or small jobs, respectively, in $\machs(\ell^*,r^*)$, e.g., $\mathtt{fracHuge} = \sum_{i\in \machs(\ell^*,r^*)} \sum_{j\in \hugejobs} x_{ij}$.
Furthermore, let $H^*$ and $L^*$ be the sets of huge and large jobs placed in $\machs(\ell^*,r^*)$ by the rounding procedure, and $k = |\machs(\ell^*,r^*)|$ the length of the interval of machines.
Now, we already established:
\begin{align}
p(S^*) + p(L^*) + p(H^*) = p_{j^*} + \mkern-12mu \sum_{i\in \machs(\ell^*,r^*)}\mkern-12mu\mathtt{load}(i) > (k-1)(1.5 - \gamma) + (2-\gamma) \label{eq:small_job_lemma_all_loads}
\end{align}
Furthermore, we have $|H^*| \leq \ceil{\mathtt{fracHuge}} \leq \ceil{\floor{ (k - p(\smalljobs(\ell^*,r^*))) / (0.5 + \xi) }} \leq (k - p(\smalljobs(\ell^*,r^*))) / (0.5 + \xi)$ due to \cref{lem:approx:hugeBound} and \cref{lp:hugeinterval} yielding:
\begin{align}
p(S^*) \leq p(\smalljobs(\ell^*,r^*)) \leq k - (0.5 + \xi)|H^*| \label{eq:small_job_lemma_smallload_vs_huge}
\end{align}
On the other hand, already the classical assignment LP constraints upper bound the load in the interval by $k$ which implies:
\[\sum_{i\in \machs(\ell^*,r^*)} \sum_{j\in \smalljobs} p_j x_{ij} \leq k - \mkern-18mu  \sum_{i\in \machs(\ell^*,r^*)} \sum_{j\in \hugejobs\cup\largejobs} p_jx_{ij} < k - 0.5\cdot\mathtt{fracLarge} - (0.5 + \xi)\cdot\mathtt{fracHuge}\]
Hence, \cref{lem:approx:hugeBound} and \cref{lem:approx:largeBound} give us:
\begin{align}
p(S^*) \leq \mkern-9mu\sum_{i\in \machs(\ell^*,r^*)} \sum_{j\in \smalljobs} p_j x_{ij} \leq k - 0.5\cdot \frac{|L^*| - 4}{2} - (0.5 + \xi)\cdot(|H^*| - 1) \label{eq:small_job_lemma_smallload_vs_huge_and_large}
\end{align}
We conclude the proof considering two cases. 
In particular, if $|L^*| \leq 6$, we have:
\begin{align*}
(2-\gamma) & \overset{(\ref{eq:small_job_lemma_all_loads})}{<} p(S^*) + p(L^*) + p(H^*) - (k-1)(1.5 - \gamma)\\
 & \overset{(\ref{eq:small_job_lemma_smallload_vs_huge})}{\leq} k - (0.5 + \xi)|H^*| + (0.5 + \xi) |L^*| + |H^*| - (k-1)(1.5 -\gamma)\\
 & = (\gamma - 0.5) (k-|H^*| - |L^*|) - \xi|H^*| + \xi|L^*|+ \gamma|H^*| + \gamma |L^*|  + 1.5 -\gamma\\
& \leq \xi|L^*| + \gamma|L^*|  + 1.5-\gamma \leq 6\xi + 6\gamma  + 1.5-\gamma \leq 2-\gamma
\end{align*}
Note that we did use $\gamma \leq 0.5$, $\gamma \leq \xi$, and $\gamma + \xi \leq \frac{1}{12}$.
If $|L^*| \geq 7$, on the other hand, we have:
\begin{align*}
(2-\gamma) & \overset{(\ref{eq:small_job_lemma_all_loads})}{<} p(S^*) + p(L^*) + p(H^*) - (k-1)(1.5 - \gamma)\\
& \overset{(\ref{eq:small_job_lemma_smallload_vs_huge_and_large})}{<} k - \frac{|L^*| - 4}{4} - (0.5 + \xi)(|H^*| - 1) + (0.5 + \xi) |L^*| + |H^*| - (k-1)(1.5 -\gamma)\\
& = (\gamma - 0.5)(k-|H^*| - |L^*|) - \frac{|L^*|}{4} -\xi|H^*| + \xi + \xi |L^*| + \gamma|H^*| + \gamma |L^*|  + 3 -\gamma \\
& \leq (\xi+ \gamma - \frac{1}{4})|L^*| + \xi + 3 -\gamma \\
& \leq (\xi+ \gamma - \frac{1}{4})7 + \xi  + 3 -\gamma = 1.25 + 8\xi + 7\gamma - \gamma \leq 2- \gamma
\end{align*}
This time, we used $\gamma \leq 0.5$, $\gamma \leq \xi$, $\xi +\gamma \leq 0.25 $ and $8\xi + 7\gamma \leq 0.75$.
Since we did reach the contradiction $(2-\gamma) < (2-\gamma)$ in both cases, the proof is complete.
\end{proof}
Lastly, we choose values for $\xi$ and $\gamma$ which satisfy all the requirements of the above lemma and maximize $\gamma$.
The biggest $\gamma$ is achieved by setting $\gamma = \xi = \frac{1}{24}$.
This concludes the proof of \Cref{the:approx2minusgamma}.

\section{Complexity Results}\label{sec:complexity}

Remember that we use the notation $[n] = \set{0,\dots,n-1}$ for each integer~$n$. 
The complexity result in this work directly build upon the ones in \cite{DBLP:conf/stacs/MaackJ20}.
In that work, a satisfiability problem denoted as \tailoredsat was introduced and shown to be NP-hard, and all reductions in the present work start from this problem.
An instance of the problem \tailoredsat is a conjunction of clauses with exactly 3 literals each.  
Each of the clauses is either a 1-in-3-clause or a 2-in-3-clause, that is, they are satisfied if exactly one or two of their literals, respectively, evaluate to \emph{true} in a given truth assignment.
We denote a $k$-in-3-clause with literals $x$, $y$, and $z$ as $(x,y,z)_k$ and the truth values true and false are denoted as $\true$ and $\false$ in the following.
There are as many 1-in-3-clauses in a \tailoredsat instance as there are 2-in-3-clauses, and, furthermore, each literal occurs exactly twice.
Hence, a minimal example for a \tailoredsat instance is given by 
$(x_0, x_1, \neg x_2)_1 \wedge (\neg x_0, x_1, x_2)_1 \wedge (x_0,\neg x_1,\neg x_2)_2 \wedge (\neg x_0, \neg x_1, x_2)_2$:
We have two 1-in-3-clauses and two 2-in-3-clauses, and two occurrences of $x_i$ and $\neg x_i$ for each $i\in [3]$.
The formula is satisfied if we map every variable to $\false$.

In each reduction, we start with an instance $I$ of \tailoredsat with $m$ many 1-in-3-clauses $C_{0},\dots,C_{m-1}$, $m$ many 2-in-3-clauses $C_{m},\dots,C_{2m-1}$ and $n$ variables $x_0,\dots,x_{n-1}$.
Since there are $2m$ clauses with $3$ literals each and 4 occurrences for each variable, we have $6m = 4n$.
In the following, the precise positions of the occurrences of the variables are important and we have to make them explicit.
To this end, let for each $j\in[n]$ and $t\in[4]$ the pair $(j,t)$ correspond to the first or second positive occurrence of variable $x_j$ if $t=0$ or $t=1$, respectively, and to the first or second negative occurrence of variable $x_j$ if $t=2$ or $t=3$.
Furthermore, let $\kappa:[n]\times [4] \rightarrow [2m]\times [3]$ be the bijection that maps $(j,t)$ to the corresponding clause index and position in that clause. 
For instance, in the above example we have $\kappa(0,2) = (1,0)$ and $\kappa(2,1) = (3,2)$.

Next, we construct an instance $I'$ of the problem considered in the respective case.
For the restricted assignment type problems, all job sizes are integral and upper bounded by some constant $T$ such that the overall size of the jobs equals $|\machs| T $.
Hence, if a machine receives jobs with overall size more or less than $T$, the objective function value is worse than $T$ for both the makespan and fair allocation case.
The goal is to show, that there is a schedule with makespan $T$ for $I'$, if and only if $I$ is a yes-instance.
This rules out approximation algorithms with rate smaller than $(T+1)/T$ for the makespan problem, and with rate smaller than $T/(T-1)$ for the fair allocation variant since the overall load is $|\machs| T $.
For the low rank problem, we first design a restricted assignment reduction using the above approach and then show that there exist low rank scheduling instances that approximate the restricted assignment instance with arbitrary precision.

In the following, we will call a schedule that assigns a load of $T$ to each machine a $T$-schedule.

\subparagraph{Simple Reduction.}

We start with a simple reduction for the general restricted assignment problem (with arbitrary restrictions) introducing several ideas and gadgets relevant for all of the following reductions.
Note that the reduction is very similar to the one by Ebenlendr et al. \cite{DBLP:journals/algorithmica/EbenlendrKS14} and to a reduction in \cite{DBLP:conf/stacs/MaackJ20}.

We have three types of basic jobs and machines, namely, truth assignment machines and jobs that are used to assign truth values to variables, clause machines and jobs that model clauses being satisfied, and variable jobs that connect the first two types:
\begin{itemize}
\item There are truth assignment machines $\truthmach(j,q)$ with $j\in[n]$ and $q\in[2]$ and one truth assignment job $\truthjob(j)$ with size $2$ and eligible on $\set{\truthmach(j,0),\truthmach(j,1)}$.
\item There are clause machines $\clausemach(i,s)$ for each $i\in[2m]$ and $s\in[3]$ and three clause jobs $\clausejob(i,s)$ each eligible on $\sett{\clausemach(i,s')}{s'\in[3]}$. 
The job $\clausejob(i,0)$ has size $1$, $\clausejob(i,2)$ has size $2$, and $\clausejob(i,1)$ has size $2$ if clause $C_i$ is a 1-in-3-clause and size $1$ otherwise.
\item Lastly, there are variable jobs $\variablejob(j,t)$ for each $j\in[n]$ and $t\in [4]$ each of size $1$ and eligible on $\set{\truthmach(j,\floor{\frac{t}{2}}), \clausemach(\kappa(j,t))}$.
\end{itemize}
First note:
\begin{claim}
The overall job size $\sum_{j\in\jobs} p(j)$ is equal to $2|\machs|$.
\end{claim}
\begin{claimproof}
There are $2n + 6m = 6n$ machines, the truth assignment jobs have overall size $2n$, the variable jobs $4n$, and the clause jobs $ 3m + 6m = 9m = 6n$.
Hence, we have $\sum_{j\in\jobs} p(j) = 12n = 2\cdot 6n = 2|\machs|$.
\end{claimproof}
Consider the case that we have a satisfying truth assignment for instance $I$.
If variable $x_j$ is assigned to $\true$, we place $\truthjob(j)$ on $\truthmach(j,0)$, $\variablejob(j,0)$ and $\variablejob(j,1)$ on $\clausemach(\kappa(j,0))$ and $\clausemach(\kappa(j,1))$, respectively, together with local size 1 clause jobs.
Furthermore, $\variablejob(j,2)$ and $\variablejob(j,3)$ are placed on $\truthmach(j,1)$ and $\clausemach(\kappa(j,2))$ and $\clausemach(\kappa(j,3))$ each receive a local size 2 clause job.
If variable $x_j$ is assigned to $\false$, we place $\truthjob(j)$ on $\truthmach(j,1)$, and the placement strategy of the positive and negative variable jobs is reversed. 
Note that the placement of the clause jobs has to work out since the truth assignment is satisfying.
This approach yields a schedule with makespan $2$.
However, it is also easy to see that a schedule with makespan $2$ yields a satisfying truth assignment by basing the assignment of $x_j$ on the placement of $\truthjob(j)$, and hence we have:
\begin{lemma}\label{lem:basic_reduction}
There is a satisfying truth assignment for $I$, if and only if there is a schedule with makespan $2$ for $I'$.
\end{lemma}
We now adapt and extend this basic reduction to all the other cases considered in this work.

\subsection{Three Ressources}\label{sec:rar3_reduction}

In the \rar{3}\ case, we can use essentially the same construction as above. 
However, the sets of eligible machines are defined using the resources and are slightly different.
The resource demands and capacities are specified in \cref{table:res3_dem_cap}.
\begin{table}
\centering
\caption{Resource demands and capacities of the jobs and machines, respectively, for the case with $3$ resources.}
\begin{tabular}{llll}
\toprule
Job/Mach. & Res. 1 & Res. 2 & Res. 3 \\ \midrule
$\truthmach(j,0)$ & $4j + 1$  & $4n - 4j$ & 1\\
$\truthmach(j,1)$ & $4j + 3$  & $4n - 4j$ & 0\\
$\truthjob(j)$ & $4j$ & $4n - 4j$ & 0\\
$\clausemach(i,s)$, $\kappa^{-1}(i,s) = (j,t)$ & $4j + t$ & $4n - (4j + t)$ & $2 + i$\\
$\clausejob(i,s)$ & 0 & 0 & $2 + i$\\
$\variablejob(j,t)$ & $4j + t$ & $4n - (4j + t)$ & $ 1 - \floor{\frac{t}{2}}$\\
\bottomrule
\end{tabular}
\label{table:res3_dem_cap}
\end{table}
It is easy to see that the choice of resources implies:
\begin{claim}
We have $\emachs(\truthjob(j)) = \set{\truthmach(j,0),\truthmach(j,1)}$ for each $j\in[n]$ and $\emachs(\variablejob(j,t)) = \set{\truthmach(j,\floor{\frac{t}{2}}), \clausemach(\kappa(j,t))}$ for each $j\in[n]$ and $t\in[4]$. 
\end{claim}
Hence, the truth assignment and variable jobs have the same sets of eligible machines as before.
For the clause jobs this is not true, however, a similar claim holds.
Remember that a $T$-schedule is a schedule in which each machine receives a load of $T$.
\begin{claim}
In any $2$-schedule each machine from $\sett{\clausemach(i,s)}{s\in[3]}$ receives exactly one job from $\sett{\clausejob(i,s)}{s\in[3]}$.
\end{claim}
\begin{claimproof}
Since each machine has to receive a load of exactly $2$, the last claim implies that each clause machine has to receive at least one clause job.
Furthermore, the clause jobs corresponding to the last clause can only be processed on the clause machines corresponding to the last clause. 
But then the analogue statement has to be true for the second to last clause and so forth. 
\end{claimproof}
Hence, \cref{lem:basic_reduction} works the same as before and we have:
\begin{theorem}
There is no better than $1.5$-approximation for \rar{3}\ and no better than $2$-approximation for the fair allocation version of this problem, unless P=NP.
\end{theorem}

\subsection{Two Ressources}\label{sec:rar2_reduction}

The reduction for \rar{2}\ is slightly more complicated.
In particular, we have three truth assignment jobs $\truthjob(j,\ell)$ with $\ell\in[3]$ and eight variable jobs $\variablejob(j,t,\circ)$ with $t\in[4]$ and $\circ\in\set{\true,\false}$ for each $j\in[n]$.
The demands, capacities and job sizes are specified in \cref{table:res2_dem_cap_size}.
\begin{table}
\centering
\caption{Resource demands and capacities of the jobs and machines, respectively, and the sizes of the jobs for the case with two resources. 
We set $\psi(\true) = 1$, $\psi(\false) = 0$, $\phi(i,0) = 0$, $\phi(i,2) = 1$, and $\phi(i,1) = k-1$ if $C_i$ is a $k$-in-3-clause.}
\begin{tabular}{llll}
\toprule
Job/Mach. & Res. 1 & Res. 2 & Size \\ \midrule
$\truthjob(j,0)$ & $2j$ & $6n - 2j$ & $1$  \\
$\truthjob(j,1)$ & $2j + 1$ & $6n - (2j + 1)$ & $1$  \\
$\truthjob(j,2)$ & $2j$ & $6n - (2j + 1)$ & $2$  \\
$\clausejob(i,s)$ & $2n + i$ & 0 & $4 + \phi(i,s)$\\
$\variablejob(j,t,\circ)$ & $2j + \floor{t/2}$ & $4j + t$ & $2 + \psi(\circ)$ \\ 
$\clausemach(i,s)$ & $2n + i$ & $4j + t$ with $(j,t) = \kappa^{-1}(i,s)$ & -\\
$\truthmach(j,q)$ & $2j + q$ & $6n - (2j + q)$ & -\\
\bottomrule
\end{tabular}
\label{table:res2_dem_cap_size}
\end{table}
First note:
\begin{claim}
The overall job size $\sum_{j\in\jobs} p(j)$ is equal to $7|\machs|$.
\end{claim}
\begin{claimproof}
There are $2n + 6m = 6n$ machines, the truth assignment jobs have overall size $4n$, the variable jobs $(2 + 3)4n = 20n$, and the clause jobs $ 12m + 15m = 27m = 18n$.
Hence, we have $\sum_{j\in\jobs} p(j) = 42n = 7\cdot 6n = 7|\machs|$.
\end{claimproof}
The choice of resources directly implies:
\begin{claim}
For each $j\in[n]$ we have $\emachs(\truthjob(j,\ell)) = \set{\truthmach(j,\ell)}$ for $\ell \in [2]$ and $\emachs(\truthjob(j,2)) = \set{\truthmach(j,0), \truthmach(j,1)}$.
\end{claim}
Regarding the clause jobs, we can employ a similar argument to the one in the last section:
\begin{claim}
In any $7$-schedule each machine from $\sett{\clausemach(i,s)}{s\in[3]}$ receives exactly one job from $\sett{\clausejob(i,s)}{s\in[3]}$.
\end{claim}
\begin{claimproof}
Any machine can process at most one clause job in a $7$-schedule due to their sizes.
Now, the three clause jobs corresponding to the last clause can only be processed on the three clause machines corresponding to the last clause due to the first resource. 
Now, we can repeat this argument for the second to last clause and so forth.
\end{claimproof}
Regarding the variable jobs, we have:
\begin{claim}
In any $7$-schedule, one of the two jobs $\variablejob(j,t,\true)$ and $\variablejob(j,t,\false)$ is assigned to $\clausemach(\kappa(j,t))$ and the other to $\truthmach(j,\floor{t/2})$.
\end{claim}
\begin{claimproof}
First note that any clause machine has to receive exactly one variable job in a 7-schedule due to the first two claims.
On the other hand, the truth assignment machines each process one truth assignment job of size $1$ and potentially another one of size $2$.
Hence, there are only two ways to reach a load of $7$ using the feasible jobs, namely with two size $2$ or two size $3$ variable jobs (combined with the second truth assignment job or not). 

We proof the claim for increasing lexicographical values of $(j,t)$ starting from $(0,0)$ and $(0,1)$.
Considering the second resource, the only variable jobs $\clausemach(\kappa(0,0))$ and $\clausemach(\kappa(0,1))$ can process are $\sett[\big]{\variablejob(0,0,\circ)}{\circ\in\set{\true, \false}}$ and $\sett[\big]{\variablejob(0,t,\circ)}{t\in[2],\circ\in\set{\true, \false}}$, respectively.
Moreover, $\truthmach(0,0)$, can only process variable jobs $\variablejob(0,t,\circ)$ with $t\in[2]$ and $\circ\in\set{\true, \false}$ but has to process either $\variablejob(0,0,\true)$ or $\variablejob(0,0,\false)$ and either $\variablejob(0,1,\true)$ or $\variablejob(0,1,\false)$ to realize a load of $7$.
Hence, the claim follows for $(j,t)\in\set{(0,0),(0,1)}$.
Now, we can repeat the same argument for $(0,2)$ and $(0,3)$, and so forth.
\end{claimproof}
Using these claims, we can show:
\begin{theorem}
There is no better than $\frac{8}{7}$-approximation for \rar{2}\ and no better than $\frac{7}{6}$-approximation for the fair allocation version of this problem, unless P=NP.
\end{theorem}
\begin{proof}
Consider the case that there is a $7$-schedule $\sigma$.
We fix a variable $x_j$ and clause $C_i$.
The claims and the jobs sizes imply that either 
\begin{align*}
\sigma^{-1}(\truthmach(j,0)) &= \set{\truthjob_{j,0},\truthjob_{j,2}, \variablejob(j,0,\false), \variablejob(j,1,\false)}\text{ and}\\
\sigma^{-1}(\truthmach(j,1)) &= \set{\truthjob_{j,1},\variablejob(j,2,\true), \variablejob(j,3,\true)},
\end{align*}
or
\begin{align*}
\sigma^{-1}(\truthmach(j,0)) &= \set{\truthjob_{j,0},\variablejob(j,0,\true), \variablejob(j,1,\true)}\text{ and}\\
\sigma^{-1}(\truthmach(j,1)) &= \set{\truthjob_{j,1},\truthjob_{j,2},\variablejob(j,2,\false), \variablejob(j,3,\false)},
\end{align*}
corresponding to the choice of assigning the value $\true$ or $\false$ to $x_j$.
Furthermore, if $C_i$ is a $k$-in-3-clause, there have to be $k$ machines out of $\sett{\clausemach(i,s)}{s\in[3]}$ that receive one big variable job (of type $\variablejob(\cdot,\cdot,\true)$) together with a small clause job, while the remaining ones receive one small variable and one big clause job.
Considering the schedule for the truth machines and the claims above, this implies that we have a satisfying truth assignment.

The other direction is the simpler one:
If we have a satisfying assignment, we choose the schedule for the truth assignment machines as above and schedule the remaining variable jobs as implied by the claims and combining the jobs correctly on the clause machines to get a $7$-schedule.

\end{proof}

\subsection{Interval Restrictions}\label{sec:rai_reduction}

In order to motivate the new ideas for the \rai reduction and to make them easier to understand, it is helpful to revisit the reduction from \cite{DBLP:conf/stacs/MaackJ20} first.
One of the main ingredients in that result is a simple trick that we will also use extensively.

\subparagraph{Pyramid Trick.}
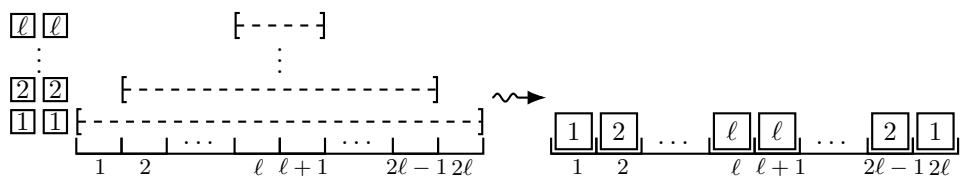
\begin{figure}
\centering
\begin{tikzpicture}[scale = 0.85]
\pgfmathsetmacro{\MachW}{0.7}
\pgfmathsetmacro{\MachH}{0.25}
\pgfmathsetmacro{\MachGap}{2.5*\MachW}
\pgfmathsetmacro{\JobGap}{0.07}
\pgfmathsetmacro{\JobWS}{0.36}
\pgfmathsetmacro{\JobW}{\MachW - 2*\JobGap}
\pgfmathsetmacro{\JobMult}{0.8}

\foreach \x/\y/\l in {	0/0/$1$,1/0/$2$,1/1/$\ell$,
						2/1/$\ell +1$,2/2/$2\ell - 1$,3/2/$2\ell$}
{
\draw[thick] (\x * \MachW + \y*\MachGap,\MachH) -- (\x * \MachW + \y*\MachGap,0) node[xshift = 0.45*\MachW cm, yshift = -0.25*\MachW cm] {\footnotesize  \l} --  (\x * \MachW + \MachW + \y*\MachGap,0) -- (\x * \MachW + \MachW +\y*\MachGap,\MachH);
}
\draw[thick] (0 * \MachW + 0*\MachGap,0) -- (3 * \MachW + 1*\MachGap,0) node[midway, above] {$\dots$};
\draw[thick] (2 * \MachW + 1*\MachGap,0) -- (3 * \MachW + 2*\MachGap,0) node[midway, above] {$\dots$};

\foreach \x/\y/\z/\l in {0/0/0/$1$,1/0/2/$2$,1/1/6/$\ell$}
{
\draw[thick, dashed, {Bracket[width = \JobWS cm]}-{Bracket[width = \JobWS cm]}] (\x * \MachW + \y*\MachGap,2*\MachH + \z*\MachH) --  (4* \MachW + 2*\MachGap - \x* \MachW - \y*\MachGap,2*\MachH + \z*\MachH);
\draw[thick] (0 * \MachW + 0*\MachGap - 2.8*\JobWS,2*\MachH + \z*\MachH - 0.5*\JobWS) rectangle (0* \MachW + 0*\MachGap - 1.8*\JobWS,2*\MachH + \z*\MachH + 0.5*\JobWS) node [midway] {\l};
\draw[thick] (0 * \MachW + 0*\MachGap - 1.4*\JobWS,2*\MachH + \z*\MachH - 0.5*\JobWS) rectangle (0* \MachW + 0*\MachGap - 0.4*\JobWS,2*\MachH + \z*\MachH + 0.5*\JobWS) node [midway] {\l};
}

\draw[->,>={Latex[length=2.5mm]},thick, line join=round,decorate, decoration={snake, segment length=6, amplitude=1.5,post=lineto, post length=9pt}] (4.2 * \MachW + 2*\MachGap,3.5*\MachH) to ++(1.2 * \MachW,0);

\node at (2 * \MachW + 1*\MachGap,6.3*\MachH) {\vdots};
\node at (- 1.6*\JobWS,6.3*\MachH) {\vdots};

\begin{scope}[xshift = 10.5 * \MachW  cm ]

\foreach \x/\y/\l/\ll in {	0/0/$1$/$1$,1/0/$2$/$2$,1/1/$\ell$/$\ell$,
						2/1/$\ell +1$/$\ell$,2/2/$2\ell - 1$/$2$,3/2/$2\ell$/$1$}
{
\draw[thick] (\x * \MachW + \y*\MachGap,\MachH) -- (\x * \MachW + \y*\MachGap,0) node[xshift = 0.5*\MachW cm, yshift = -0.25*\MachW cm] {\footnotesize  \l} --  (\x * \MachW + \MachW + \y*\MachGap,0) -- (\x * \MachW + \MachW +\y*\MachGap,\MachH);
\draw[thick] (\JobGap + \x*\MachW + \y*\MachGap,\JobGap) rectangle (\JobGap + \JobW + \x*\MachW + \y*\MachGap,\JobGap + \JobW) node [midway] {\ll};
}

\draw[thick] (0 * \MachW + 0*\MachGap,0) -- (3 * \MachW + 1*\MachGap,0) node[midway, above] {$\dots$};
\draw[thick] (2 * \MachW + 1*\MachGap,0) -- (3 * \MachW + 2*\MachGap,0) node[midway, above] {$\dots$};

\end{scope}

\end{tikzpicture}
\caption{
A visualization of the pyramid trick.
Squares represent jobs and the intervals in brackets next to them their sets of eligible machines.
} 
\label{fig:pyramid_trick}
\end{figure}
Consider the case depicted in \cref{fig:pyramid_trick}. 
We have $2\ell$ consecutive machines and $\ell$ pairs of jobs.
The $i$-th pair of jobs is eligible on the $i$-th machine and up to and including the $(2\ell + 1 -i)$-th machine.
Furthermore, we assume that for some reason each machine has to receive at least one of the jobs.
Then the first and last machine each have to receive one job from the first pair because there are no other eligible jobs that can be processed on these machines.
Now, the same argument can be repeated for the second and second to last machine and so on.
Hence, machine $i$ and $(2\ell + 1 -i)$ each have to receive exactly one job from pair $i$.

\subparagraph{Sorting.}

Next, consider that in the ordering of the machines the truth assignment machines are placed on the left and the clause machines on the right.
We could use similar truth assignment and clause jobs as in the reduction in the beginning of this chapter.
However, variable jobs each are eligible on one truth assignment and one clause machine.
Hence they have to be eligible on all machines in between in a naive adaptation of the reduction to the interval case.
If we want to use the pyramid trick to deal with this problem, then intuitively decisions regarding variables in later clauses have to be made before the decisions for variables in earlier clauses and this, of course, cannot be guaranteed regardless of the fixed order of the clauses or variables.
The main work in \cite{DBLP:conf/stacs/MaackJ20} was to remedy this situation by -- roughly speaking -- sorting the information regarding the variables made in the truth assignment gadget to enable the use of the pyramid trick.
To do so several gadgets have been introduced that were intertwined with the truth assignment gadget and carefully build up the ordered information using the pyramid trick and interlocking job sizes.
In particular, private loads can be introduced, that is, jobs that may only be processed by one particular machine.
Using these private loads, it can be guaranteed that the machine has to receive certain types of jobs in a $T$-schedule.
The problem with this argument is that the job sizes can get big rather fast if too many different job types are eligible on the same machines resulting in a high value for $T$.
Now, the main idea in the present work is to decouple the decision and the sorting process and make the sorting process as simple as possible to enable smaller job sizes and therefore a stronger result. 
Curiously, the sorting process in \cite{DBLP:conf/stacs/MaackJ20} could be interpreted as some variant of insertion sort, while the one used in the present reduction resembles bubble sort.

\subparagraph{Machines and Order.}

Let $k\in\Oh(n^2)$ be a parameter to be specified later in this paragraph.
In addition to the truth assignment and clause machines, we introduce the following ones:
\begin{itemize}
\item For each $j\in[n]$ and $t\in[4]$ there are two gateway machines: one forward $\fgatemach(j,t)$ and one backward gateway machine $\bgatemach(j,t)$.
\item For each $\ell\in[k]$, $j\in[n]$, and $t\in[4]$ there are two sorting machines: one forward $\fsortmach(\ell,j,t)$ and one backward sorting machine $\bsortmach(\ell,j,t)$.
\end{itemize}
Let $\mathcal{T}$, $\mathcal{G}$, and $\mathcal{C}$ be the sets of truth assignment, gateway, and clause machines, respectively.
Moreover, for each $\ell\in[k]$ let $\mathcal{S}_\ell = \sett{\bsortmach(\ell,j,t),\fsortmach(\ell,j,t)}{j\in[n], t\in[4]}$ be the $\ell$-th set of sorting machines.
We define the overall order of the machines by setting an internal order for each set of machines as well as an order of the machine sets. 
However, before we can do so, we need some additional concepts and notation.
In particular, let $\varphi_{0}$ be the sequence of $(j,t)$-pairs with $j\in[n]$, $t\in[4]$ with increasing lexicographical order and $\psi_{0} = \kappa(\varphi_{0})$, i.e., $\varphi_{0} = ((0,0),\dots,(0,3),\dots,(n-1,0),\dots,(n-1,3))$ and $\psi_{0} = (\kappa(0,0),\dots,\kappa(0,3),\dots,\kappa(n-1,0),\dots,\kappa(n-1,3))$.
Hence, $\psi_{0}$ is a permutation of the pairs $(i,s)$ with $i\in[2m]$ and $s\in[3]$.
We consider sorting $\psi_{0}$ with the goal of reaching the increasing lexicographical order.
Let $k$ be the number of transpositions performed by bubble sort if we do this.
Furthermore, let $\psi_{\ell+1}$ for $\ell\in[k]$ be the sequence we get after the first $(\ell + 1)$-transpositions and $\varphi_{\ell+1} = \kappa^{-1}(\psi_{\ell+1})$.
The use of bubble sort guarantees that $k\in\Oh(n^2)$ and that two consecutive sequences $\varphi_{\ell}$, $\varphi_{\ell+1}$ differ only by two consecutive entries that are transposed. 
For any finite sequence $\chi$, we denote the reversed sequence as $\bar{\chi}$.   
Now, the ordering is specified as follows: 
\begin{itemize}
\item The sets are ordered as follows: $\mathcal{T},\mathcal{G}, \mathcal{S}_0,\dots,\mathcal{S}_{k-1},\mathcal{C}$.
\item The truth assignment machines are ordered in \emph{increasing} lexicographical order of the indices $(j,q)$.
\item The clause machines are ordered in \emph{decreasing} lexicographical order of the indices $(i,s)$.
\item The backward gateway machines are placed before the forward gateway machines and for each $\ell\in[k]$ the backward sorting machines are placed before the forward sorting machines from $\mathcal{S}_\ell$ as well. 
\item The forward and backward gateway machines are ordered in increasing and decreasing lexicographical order of the indices $(j,t)$, i.e., $\varphi_0$ and $\bar{\varphi}_0$, respectively.
\item For each $\ell\in[k]$, the backward sorting machines are sorted according to the placement of the $(j,t)$ indices in $\bar{\varphi}_{\ell}$.
\item For each $\ell\in[k]$, the forward sorting machines are sorted according to the placement of the $(j,t)$ indices in $\varphi_{\ell+1}$.
\end{itemize}
The peculiar ordering of the machines is designed to enable the use of the pyramid trick.

\subparagraph{Jobs, Sizes, and Eligibilities.}
We give a full list of all jobs together with their sizes and define the sets of eligible machines by stating the respective first and last eligible machine for each job.
We will need one more definition: 
Let $\xi:[k] \rightarrow [n]\times [4]$ be the function that maps $\ell$ to the distinct pair $(j,t)$ that has a higher index in $\varphi_{\ell+1}$ than in $\varphi_{\ell}$.
\begin{itemize}
\item Truth assignment jobs: For each $j\in[n]$ there is a job $\truthjob(j)$ with size $2$, first machine $\truthmach(j,0)$ and last machine $\truthmach(j,1)$.
\item Variable jobs: For each $j\in[n]$, $t\in[4]$, and $\circ\in\set{\true, \false}$ there is a job $\variablejob(j,t,\circ)$ with size $2$ if $\circ = \false$ and $3$ otherwise, first machine $\truthmach(j,\floor{t/2})$ and last machine $\bgatemach(j,t)$.
\item Gateway jobs: For each $j\in[n]$, $t\in[4]$, and $\circ\in\set{\true, \false}$ there is a job $\gatejob(j,t,\circ)$ with size $5$ if $\circ = \false$ and $4$ otherwise, first machine $\bgatemach(j,t)$ and last machine $\fgatemach(j,t)$.
\item Bridge jobs: For each $\ell\in[k+1]$, $j\in[n]$, $t\in[4]$, and $\circ\in\set{\true, \false}$ there is a bridge job $\bridgejob(\ell,j,t,\circ)$ with size $1$ if $\circ = \false$ and $2$ otherwise, first machine either $\fgatemach(j,t)$ if $\ell =0$ or $\fsortmach(\ell - 1,j,t)$ otherwise, and last machine either $\bsortmach(\ell, j,t)$ if $\ell<k$ or $\clausemach(\kappa(j,t))$ otherwise.
\item Sorting jobs: For each $\ell\in[k]$, $j\in[n]$, $t\in[4]$, and $\circ\in\set{\true, \false}$ there is a job $\sortjob(\ell,j,t,\circ)$. 
If $\xi(\ell)=(j,t)$, it has size $4$ if $\circ = \false$ and $3$ otherwise, and, if $\xi(\ell)\neq(j,t)$, it has size $7$ if $\circ = \false$ and $6$ otherwise.
The first machine of $\sortjob(\ell,j,t,\circ)$ is $\bsortmach(\ell,j,t)$ and the last machine is $\fsortmach(\ell,j,t)$.
\item Clause jobs: For each $i\in[2m]$ and $s\in[3]$ there is a job $\clausejob(i,s)$ with size $7$ if $s = 1$, $8 - k$ if $s= 2$  and $C_i$ is a $k$-in-3-clause, and $6$ if $s=3$, first machine $\clausemach(i,2)$ and last machine $\clausemach(i,0)$. 
\item Private loads: Each truth assignment machine has a private load (a job eligible only one one machine) of $2$, each backward or forward gateway machine a load of $1$ or $2$, respectively, and for each $\ell\in[k]$ the sorting machines $\bsortmach(\ell,\xi(\ell))$ and $\fsortmach(\ell,\xi(\ell))$ have a private load of $3$.
\end{itemize}
This concludes the formal description of the reduction and we give a brief high level discussion.
We have a truth assignment gadget that determines the truth values of the variables and is followed by the gateway gadget, whose sole purpose is to decouple the used job sizes in the truth assignment gadget and the sorting gadget.
Next there is the sorting gadget that slowly reorders the information about the decisions in the truth assignment gadget, and lastly there is the clause gadget in which the truth assignment is evaluated.
The connection between the truth assignment and gateway gadget is provided by the variable jobs and all other connections are realized via bridge jobs.
A sketch of the overall structure of the reduction is provided in \cref{fig:rai_reduction_structure}.
\begin{figure}
\centering
\begin{tikzpicture}[>={Latex[length=1.6mm]}]
\pgfmathsetmacro{\w}{0.6}
\pgfmathsetmacro{\mgap}{0.15*\w}
\pgfmathsetmacro{\mh}{0.2*\w}
\pgfmathsetmacro{\ph}{0.2*\w}
\pgfmathsetmacro{\mjgap}{0.6}
\pgfmathsetmacro{\jgap}{0.1}
\pgfmathsetmacro{\js}{0.04}
\pgfmathsetmacro{\iw}{0.2*\w}

\pgfmathsetmacro{\mw}{0.3}
\pgfmathsetmacro{\gw}{1}
\pgfmathsetmacro{\bbh}{0.5}
\pgfmathsetmacro{\bsh}{0.3}
\pgfmathsetmacro{\ih}{0.3}
\pgfmathsetmacro{\ph}{0.22}
\pgfmathsetmacro{\sh}{0.15}

\foreach \x/\g/\y/\l in {0/0/1/$\mathcal{T}$, 1/1/4/$\mathcal{G}$, 5/2/4/$\mathcal{S}_0$, 13/4/4/$\mathcal{S}_{k-1}$, 17/5/2/$\mathcal{C}$}{
\draw[thick] (\x*\w + \g*\mgap,0) -- (\x*\w + \g*\mgap + \y*\w,0) node [midway, yshift = -0.4*\w cm] {\l};
}
\foreach \x/\g/\y/\p in {0/0/1/->, 1/1/2/<-, 3/1/2/->, 5/2/2/<-, 7/2/2/->, 13/4/2/<-, 15/4/2/->, 17/5/2/<-}{
\draw[thick, \p, shorten <= 1pt, shorten >= 1pt] (\x*\w + \g*\mgap,\ph) -- (\x*\w + \g*\mgap + \y*\w,\ph);
\draw[thick] (\x*\w + \g*\mgap,\mh) --  (\x*\w + \g*\mgap,0) -- (\x*\w + \g*\mgap + 0.5*\w,0);
\draw[thick] (\x*\w + \g*\mgap + \y*\w,\mh) -- (\x*\w + \g*\mgap + \y*\w,0) -- (\x*\w + \g*\mgap + \y*\w - 0.5*\w,0);
}

\node[] at (11*\w + 3*\mgap,0) {$\dots$};



\foreach \z/\m/\mm in {0/0/0, 1/0/1, 2/1/2, 3/1/3, 6/4/8, 7/4/9}{
\draw[thick, color= mdarkblue!90!black] (0*\w + 0*\mgap + \js + \m*\iw,\mjgap + \z*\jgap) -- ++(\mgap + 3*\w - 2*\js -\m*\iw-\mm*\iw,0);
}

\draw[densely dotted, thick, color = mdarkblue!90!black] (1*\w + 0.5*\mgap, \mjgap + 3.6*\jgap) -- + (0,1.5*\jgap);
\node at (1.5*\w + 0.5*\mgap, 0.7*\mjgap) {\textcolor{mdarkblue!90!black}{\scriptsize $\variablejob(*,*,*)$}};

\foreach \z/\m in {0/0, 1/1, 4/8, 5/9}{
\draw[thick, color= mred] (1*\w + 1*\mgap + \js + \m*\iw,\mjgap + 9*\jgap - \z*\jgap) -- ++(4*\w - 2*\js -2*\m*\iw,0);
}

\draw[densely dotted, thick, color = mred] (3*\w + 1.0*\mgap, \mjgap + 5.6*\jgap) -- + (0,1.5*\jgap);
\node at (3.2*\w + 0.5*\mgap, 1.3*\mjgap + 9*\jgap) {\textcolor{mred}{\scriptsize $\gatejob(*,*,*)$}};

\foreach \z/\m in {0/0, 1/1, 4/8, 5/9}{
\draw[thick, color= mblue!90!black] (3*\w + 1*\mgap + \js + \m*\iw,\mjgap + 2*\jgap + \z*\jgap) -- ++(\mgap + 4*\w - 2*\js -2*\m*\iw,0);
}

\draw[densely dotted, thick, color = mblue!90!black] (5*\w + 1.5*\mgap, \mjgap + 3.6*\jgap) -- + (0,1.5*\jgap);
\node at (5.2*\w + 0.5*\mgap, 0.7*\mjgap + 2*\jgap) {\textcolor{mblue!90!black}{\scriptsize $\bridgejob(0,*,*,*)$}};

\foreach \z/\m in {0/0, 1/1, 4/8, 5/9}{
\draw[thick, color= mdarkorange] (5*\w + 2*\mgap + \js + \m*\iw,\mjgap + 9*\jgap - \z*\jgap) -- ++(4*\w - 2*\js -2*\m*\iw,0);
}

\draw[densely dotted, thick, color = mdarkorange] (7*\w + 2.0*\mgap, \mjgap + 5.6*\jgap) -- + (0,1.5*\jgap);
\node at (7.2*\w + 0.5*\mgap, 1.3*\mjgap + 9*\jgap) {\textcolor{mdarkorange}{\scriptsize $\sortjob(0,*,*,*)$}};

\foreach \z/\m in {0/0, 1/1, 4/8, 5/9}{
\draw[thick, color= mblue!90!black] (7*\w + 2*\mgap + \js + \m*\iw,\mjgap + 2*\jgap + \z*\jgap) -- ++(4*\w - 2*\js -2*\m*\iw,0);
}

\draw[densely dotted, thick, color = mblue!90!black] (9*\w + 2.0*\mgap, \mjgap + 3.6*\jgap) -- + (0,1.5*\jgap);
\node at (9.2*\w + 0.5*\mgap, 0.7*\mjgap + 2*\jgap) {\textcolor{mblue!90!black}{\scriptsize $\bridgejob(1,*,*,*)$}};

\foreach \z/\m in {0/0, 1/1, 4/8, 5/9}{
\draw[thick, color= mblue!90!black] (11*\w + 4*\mgap + \js + \m*\iw,\mjgap + 2*\jgap + \z*\jgap) -- ++(4*\w - 2*\js -2*\m*\iw,0);
}

\draw[densely dotted, thick, color = mblue!90!black] (13*\w + 4.0*\mgap, \mjgap + 3.6*\jgap) -- + (0,1.5*\jgap);
\node at (13.2*\w + 2.7*\mgap, 0.7*\mjgap + 2*\jgap) {\textcolor{mblue!90!black}{\scriptsize $\bridgejob(k-1,*,*,*)$}};

\foreach \z/\m in {0/0, 1/1, 4/8, 5/9}{
\draw[thick, color= mdarkorange] (13*\w + 4*\mgap + \js + \m*\iw,\mjgap + 9*\jgap - \z*\jgap) -- ++(4*\w - 2*\js -2*\m*\iw,0);
}

\draw[densely dotted, thick, color = mdarkorange] (15*\w + 4.0*\mgap, \mjgap + 5.6*\jgap) -- + (0,1.5*\jgap);
\node at (15.2*\w + 3.0*\mgap, 1.3*\mjgap + 9*\jgap) {\textcolor{mdarkorange}{\scriptsize $\sortjob(k-1,*,*,*)$}};

\foreach \z/\m in {0/0, 1/1, 4/8, 5/9}{
\draw[thick, color= mblue!90!black] (15*\w + 4*\mgap + \js + \m*\iw,\mjgap + 2*\jgap + \z*\jgap) -- ++(\mgap + 4*\w - 2*\js -2*\m*\iw,0);
}

\draw[densely dotted, thick, color = mblue!90!black] (17*\w + 4.5*\mgap, \mjgap + 3.6*\jgap) -- + (0,1.5*\jgap);
\node at (17.2*\w + 3.2*\mgap, 0.7*\mjgap + 2*\jgap) {\textcolor{mblue!90!black}{\scriptsize $\bridgejob(k,*,*,*)$}};

\end{tikzpicture}
\caption{
A visualization of the job and machine structure for the reduction regarding \rai.
The lower part corresponds to the sets of machines with arrows corresponding to the direction of their ordering; and the upper part to the different job sets with lines corresponding to the intervals of eligible machines for pairs of jobs.
Truth assignment and clause jobs as well as private loads are not depicted.
} 
\label{fig:rai_reduction_structure}
\end{figure}
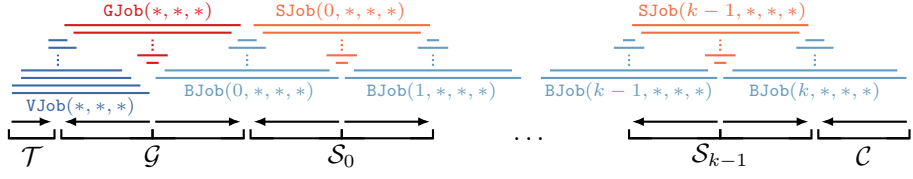

\subparagraph{Analysis.}

The following claim can be easily verified by counting the different machines and jobs and adding up the corresponding sizes.
\begin{claim}
The overall job size $\sum_{j\in\jobs} p(j)$ is equal to $8|\machs|$.
\end{claim}
\begin{claimproof}
There are $2n + 8n + 8kn + 6m = (14 + 8k)n$ machines. 
On the other hand, the truth assignment jobs have overall size $2n$, the variable jobs $20n$, the gateway jobs $36n$, the bridge jobs $12(k+1)n$, the sorting jobs $13k(4n - 1) + 7k$, the clause jobs $ 13\cdot3m = 26n$, and the private loads $4n + 12n + 6k$.
Hence, we have $\sum_{j\in\jobs} p(j) = (112 + 64k)n = 8 (14+8k)n = 8|\machs|$.
\end{claimproof}
Next, we consider the placement of clause jobs:
\begin{claim}\label{claim:rai_clause_jobs}
In any $8$-schedule each machine from $\sett{\clausemach(i,s)}{s\in[3]}$ receives exactly one job from $\sett{\clausejob(i,s)}{s\in[3]}$.
\end{claim}
\begin{claimproof}
The jobs three jobs from $\sett{\clausejob(i,s)}{s\in[3]}$ are eligible only on the three machines from $\sett{\clausemach(i,s)}{s\in[3]}$.
Since they each have a size greater than $4$, only one of them can be placed on each machine in an $8$-schedule.
This implies the claim.
\end{claimproof}
The following claims consider the placement of sorting, gateway, bridge and variable jobs.
In each of these claims the pyramid trick is used although there is a slight variation in the case of the variable jobs.
\begin{claim}\label{claim:rai_sorting_jobs}
Let $\ell\in[k]$, $j\in[n]$, and $t\in[4]$.
In any $8$-schedule, the sorting jobs $\sortjob(\ell,j,t,\true)$ and $\sortjob(\ell,j,t,\false)$ are assigned to their first or last eligible machine, i.e., $\bsortmach(\ell,j,t)$ or $\fsortmach(\ell,j,t)$, and each of the two machines receives exactly one of the two jobs. 
\end{claim}
\begin{claimproof}
For a fixed index $\ell$, the sorting jobs from $ \sett{\sortjob(\ell,j,t,\circ)}{j\in[n],t\in[4],\circ\in\set{\true, \false}}$ are only eligible on $\mathcal{S}_\ell = \sett{\bsortmach(\ell,j,t),\fsortmach(\ell,j,t)}{j\in[n], t\in[4]}$. 
All but two of these jobs have size greater than $5$.
Moreover, two of the sorting machines from $\mathcal{S}_\ell$ have a private load of $3$ and the remaining two sorting jobs have a size of at least $3$.
Hence, in any $8$-schedule, each machine from $\mathcal{S}_\ell$ receives exactly one of these sorting jobs and the two machine with private loads have to receive the two smaller ones.
In particular, let $(j^*,t^*) = \xi(\ell)$. 
Then the two smaller sorting jobs are $\sortjob(\ell,j^*,t^*,\true)$ and $\sortjob(\ell,j^*,t^*,\false)$ and they have to be placed on $\bsortmach(\ell,j^*,t^*)$ or $\fsortmach(\ell,j^*,t^*)$.
Now, for $\mathcal{S}_\ell\setminus\set{\bsortmach(\ell,j^*,t^*), \fsortmach(\ell,j^*,t^*)}$ and $ \sett{\sortjob(\ell,j,t,\circ)}{j\in[n],t\in[4],\circ\in\set{\true, \false}} \setminus \set{\sortjob(\ell,j^*,t^*,\true),\sortjob(\ell,j^*,t^*,\false)}$ we can employ the pyramid trick, i.e., on the first and last machines from this set there are only two eligible sorting jobs and hence each of them has to receive one of the two. 
This argument can be iterated thus proving the claim.
\end{claimproof}

\begin{claim}\label{claim:rai_gateway_jobs}
Let $j\in[n]$ and $t\in[4]$.
In any $8$-schedule, the gateway jobs $\gatejob(j,t,\true)$ and $\gatejob(j,t,\false)$ are assigned to their first or last eligible machine, i.e., $\bgatemach(j,t)$ or $\fgatemach(j,t)$, and each of the two machines receives exactly one of the two jobs. 
\end{claim}
\begin{claimproof}
The gateway jobs are eligible only on gateway machines, there are as many gateway jobs as gateway machines, the gateway machines each have a private load of at least $1$, and the gateway jobs have a size of at least $4$. 
Hence, in any $8$-schedule each gateway machine has to receive exactly one gateway job.
Now, considering the eligibilities of the jobs and the ordering of the machines, the claim is implied by the pyramid trick. 
\end{claimproof}

\begin{claim}\label{claim:rai_bridge_jobs}
Let $\ell\in[k+1]$, $j\in[n]$, $t\in[4]$, $\mathtt{XM} = \fgatemach(j,t)$ if $\ell =0$ and $\mathtt{XM} =\fsortmach(\ell-1,j,t)$ otherwise, and $\mathtt{YM} = \bsortmach(\ell, j,t)$ if $\ell<k$ and $\mathtt{YM} = \clausemach(\kappa(j,t))$ otherwise.
In any $8$-schedule, the bridge jobs $\bridgejob(\ell,j,t,\true)$ and $\bridgejob(\ell,j,t,\false)$ are assigned to their first or last eligible machine, i.e., $\mathtt{XM}$ or $\mathtt{YM}$, and each of the two machines receives exactly one of the two jobs. 
\end{claim}
\begin{claimproof}
We first consider the case $\ell\in\set{1,\dots, k-1}$.
The bridge jobs from $\sett{\bridgejob(\ell,j,t,\circ)}{j\in[n],t\in[4],\circ\in\set{\true, \false}}$ are only eligible on machines belonging to the set $\sett{\fsortmach(\ell-1,j,t),\bsortmach(\ell,j,t)}{j\in[n], t\in[4]}$ and there are as many jobs as there are machines. 
Furthermore, due to the placement of the sorting jobs established in \cref{claim:rai_sorting_jobs} and the placement of private loads, we know that each of the respective machines receives a load of at least $6$ and at most $7$ in any $8$-schedule due to sorting jobs or private loads.
Hence, in an $8$-schedule the sorting machines from $\sett{\fsortmach(\ell-1,j,t),\bsortmach(\ell,j,t)}{j\in[n], t\in[4]}$ have to receive at least one more job and the only other eligible jobs are the respective bridge jobs.
Therefore each of the respective machines receives exactly one bridge job from $\sett{\bridgejob(\ell,j,t,\circ)}{j\in[n],t\in[4],\circ\in\set{\true, \false}}$.
Now, the pyramid trick implies the claim.
The cases $\ell=0$ and $\ell=k$ work analogously and additionally make use of \cref{claim:rai_gateway_jobs} and \cref{claim:rai_clause_jobs}, respectively. 
\end{claimproof}
The next two claims are proved together.
\begin{claim}\label{claim:rai_variable_jobs}
Let $j\in[n]$ and $t\in[4]$.
In any $8$-schedule, the variable jobs $\variablejob(j,t,\true)$ and $\variablejob(j,t,\false)$ are assigned to their first or last eligible machine, i.e., $\truthmach(j,\floor{t/2})$ or $\bgatemach(j,t)$, and each of the two machines receives exactly one of the two jobs. 
\end{claim}
\begin{claim}\label{claim:rai_coherent_variable_jobs}
Let $j\in[n]$.
In any $8$-schedule $\sigma$ there are two possibilities regarding the schedule for the truth assignment machines:
\begin{enumerate}
\item $\truthmach(j,0)$ receives the jobs from $\set{\variablejob(j,0,\true),\variablejob(j,1,\true)}$ and $\truthmach(j,1)$ receives $\set{\variablejob(j,2,\false),\variablejob(j,3,\false), \truthjob(j)}$.
\item $\truthmach(j,0)$ receives the jobs from $\set{\variablejob(j,0,\false),\variablejob(j,1,\false), \truthjob(j)}$ and $\truthmach(j,1)$ receives $\set{\variablejob(j,2,\true),\variablejob(j,3,\true)}$.
\end{enumerate}
\end{claim}
\begin{claimproof}
First note that truth assignment machines each have a private load of $2$ and each job eligible on such a machine has a size of either $2$ or $3$.
Hence, in an $8$-schedule, these machines have to receive either two jobs of size $3$ or three jobs of size $2$.
Furthermore, we know that each backward gateway machine has a private load of $1$ and receives exactly one size $4$ or $5$ gateway job in any $8$-schedule (see \cref{claim:rai_gateway_jobs}).
Since variable jobs are the only other jobs eligible on backward gateway machines, each of them has to receive exactly one such job.
The remainder of the proof is a slight variation of the pyramid trick. 
Consider the first pair $\set{\truthmach(0,0),\truthmach(0,1)}$ of truth assignment machines in the ordering.
The eligible jobs on these machines are the truth assignment job $\truthjob(0)$ with size $2$ and the variable jobs $\sett{\variablejob(0,t,\circ)}{t\in[4],\circ\in\set{\true,\false}}$ with size $2$ or $3$.
Of these, the truth assignment job is only eligible on the two respective machines while the variable jobs are the only variable jobs eligible on the last four backward gateway machines $\sett{\bgatemach(0,t)}{t\in[4]}$ in the ordering.
This already implies that one of the two variable machines has to process two of the size $2$ variable jobs along with the truth assignment job while the other truth assignment machine receives two size $3$ variable jobs. 
Moreover, each of the gateway machines has to process exactly one of the variable jobs.
The variable jobs are in fact more restricted, i.e., on the first variable machine only the jobs from $\sett{\variablejob(0,t,\circ)}{t\in[2],\circ\in\set{\true,\false}}$ are eligible and on the backwards gateway machine $\bgatemach(0,t)$ the only eligible variable jobs are included in $\sett{\variablejob(0,t',\circ)}{t'\in[t+1],\circ\in\set{\true,\false}}$.
Hence, there are only two possible placements of the respective variable jobs:
\begin{enumerate}
\item $\truthmach(0,0)$ receives the jobs from $\set{\variablejob(0,0,\true),\variablejob(0,1,\true)}$, $\truthmach(0,1)$ receives $\set{\variablejob(0,2,\false),\variablejob(0,3,\false), \truthjob(0)}$, and $\bgatemach
(0,t)$ receives $\variablejob(0,t,\circ_t)$ for each $t\in[4]$ with $\circ_t = \false$ if $t\in[2]$ and $\circ_t = \true$ otherwise.
\item $\truthmach(0,0)$ receives the jobs from $\set{\variablejob(0,0,\false),\variablejob(0,1,\false), \truthjob(0)}$, $\truthmach(0,1)$ receives $\set{\variablejob(0,2,\true),\variablejob(0,3,\true)}$, and $\bgatemach
(0,t)$ receives $\variablejob(0,t,\circ_t)$ for each $t\in[4]$ with $\circ_t = \true$ if $t\in[2]$ and $\circ_t = \false$ otherwise.
\end{enumerate}
This proves the claims for $j=0$ and the argument can be iterated for $j>0$.
\end{claimproof}

\begin{theorem}
There is no better than $\frac{9}{8}$-approximation for \rai and no better than $\frac{8}{7}$-approximation for the fair allocation version of this problem, unless P=NP.
\end{theorem}
\begin{proof}
Consider the case that there is an $8$-schedule. 
For a given variable variable $x_j$, we consider the schedule for the corresponding variable jobs described in \cref{claim:rai_coherent_variable_jobs}.
There are two possibilities in the claim.
If the first is true, we assign $x_j$ to $\false$, and if the second is true, we assign $x_j$ to $\false$. 
Now note that the Claims \ref{claim:rai_clause_jobs}-\ref{claim:rai_variable_jobs} and the sizes of the jobs imply that if for fixed $j\in[n]$, $t\in[4]$ and $\circ\in\true$ the job $\variablejob(j,t,\circ)$ is assigned to its last eligible machine than the same is true for each job from $\set{\gatejob(j,t,\circ)}\cup\sett{\bridgejob(\ell,j,t,\circ)}{\ell\in[k+1]}\cup\sett{\sortjob(\ell,j,t,\circ)}{\ell\in[k]}$.
This can be seen via a simple induction and the fact that the sizes of jobs in a job pair always differ by one.
Furthermore, for each clause $C_i$ the sizes of the jobs that the number of bridge jobs with value $\true$ in the third component is exactly $k$ if $C_i$ is a $k$-in-$3$-clause.
Hence, the truth assignment is satisfying.

If, on the other hand, we have a satisfying truth assignment, we choose the first possibility in \cref{claim:rai_coherent_variable_jobs} for the schedule of the truth assignment and variable jobs if a variable $x_j$ is assigned to $\false$ and the second one if it is assigned to $\true$ the schedule of the remaining jobs.
Now the schedule of the remaining jobs can be chosen according to the claims aiming for an $8$-schedule.
\end{proof}

\subsection{Rank Three}\label{sec:rank3}

The result for \lrs{3} is structured as follows.
We first present another restricted assignment reduction reusing several of the ideas presented so far. 
Then we show that the corresponding job sizes and eligibilities can be simulated with arbitrary precision by \lrs{3} instances.
The basic idea of the reduction is again to make decisions in a truth assignment gadget and verify these decisions using a clause gadget.
These two gadgets will be very similar to the ones described in the very first reduction for restricted assignment. 
However, similar to the \rai\ reduction, we will maintain an order of the machines corresponding to decreasing machine speeds in the first dimension of the speed vectors of the machines.
The main idea is now to use the other two dimensions to realize eligibility for the jobs transporting the information regarding the assignment decisions using increasing machine speeds in the second and third dimension.
Like in the \rai\ case, this approach is complicated by the fact that the variables can occur arbitrarily in the clauses preventing compatible orderings of the machines in the two gadgets.
Hence, we spend a lot of effort to sort the machines.
In particular, we will define a sequence of blocks of machines with jobs eligible either only in one block or a pair of blocks directly succeeding each other.
The first block will correspond to the truth assignment, the last to the clause gadget, and the remaining blocks contain one machine for each job occurrence with the ordering changing little by little.
Since the actual speed and size vectors of the machines and jobs, respectively, will not be discussed until later, the reasons for many of the peculiarities of the reduction will not become clear until then.

\subparagraph{Machines and Order.}

We reuse notation from \cref{sec:rai_reduction}.
In particular, we use mostly the same notation as in the respective paragraph regarding machines and order, i.e., let $\varphi_{0}$ be the sequence of $(j,t)$-pairs with $j\in[n]$, $t\in[4]$ with increasing lexicographical order and $\psi_{0} = \kappa(\varphi_{0})$.
Moreover, let $k\in\Oh(n^2)$ be the number of transpositions performed by bubble sort if we sort $\psi_{0}$ with the goal of reaching the increasing lexicographical order, and let $\psi_{\ell+1}$ for $\ell\in[k]$ be the sequence we get after the first $(\ell + 1)$-transpositions and $\varphi_{\ell+1} = \kappa^{-1}(\psi_{\ell+1})$.

We now first define the machines and then the machine blocks.
In particular, we have:
\begin{itemize}

\item Truth assignment machines $\truthmach(j,q)$ for $j\in[n]$ and $q\in[2]$.

\item Sorting machines $\sortmach(\ell,q,j,t)$ for each $\ell \in[k]$, $q\in [3]$, $j\in[n]$, and $t\in[4]$. 

\item Amplifier machines $\ampmach(\ell,q)$ for each $\ell \in[k]$ and $q\in [3]$.

\item Clause machines $\clausemach(i,s)$ for each $i\in[2m]$ and $s\in[3]$.
\end{itemize}
We partition the machines into a sequence of $3k + 2$ blocks.
The first is the truth assignment block $\tblock = \sett{\truthmach(j,q)}{j\in[n],q\in[2]}$, then we have a sorting block $\sblock_{3\ell + q} = \sett{\sortmach(\ell,q,j,t)}{j\in[n],t\in[4]}\cup \set{\ampmach(\ell,q)}$ for each $\ell\in[k]$ and $q\in[3]$ ordered increasingly by index, and lastly the clause block $\cblock =  \sett{\clausemach(i,s)}{i\in[2m],s\in[3]}$.
One sorting step while be carried out in a triple of succeeding sorting blocks, i.e., $\sblock_{3\ell}\cup \sblock_{3\ell + 1} \cup \sblock_{3\ell + 2}$ facilitate the transition from $\varphi_{\ell}$ to $\varphi_{\ell+1}$.

\subparagraph{Jobs, Sizes, and Eligibilities.}

For each $\ell\in[k]$, let $(j^{>}_\ell,t^{>}_\ell)$ and $(j^{<}_\ell,t^{<}_\ell)$ be the pairs with incremented and decremented index, respectively, when comparing $\varphi_{\ell}$ with $\varphi_{\ell+1}$. 
We define all jobs, sizes, and eligibilities, and furthermore record to which machine blocks the respective eligible machines belong.
In particular, we have:
\begin{itemize}

\item One truth assignment job $\truthjob(j)$ for each $j\in[n]$ eligible on $\set{\truthmach(j,0),\truthmach(j,1)}\subseteq \tblock$ and with size $2$.

\item Variable jobs $\variablejob(j,t)$ for each $j\in[n]$ and $t\in [4]$ each of size $1$ and eligible on $\set{\truthmach(j,\floor{\frac{t}{2}}), \sortmach(0,0,j,t)} \subseteq \tblock \cup \sblock_{0}$.

\item Sorting jobs $\sortjob(\ell,q,j,t)$ for each $\ell \in[k]$, $q\in [3]$, $j\in[n]$, and $t\in[4]$. 
Most of these jobs have size $1$ and are eligible on $\set{\sortmach(\ell,q,j,t),\sortmach(\ell + \floor{q/2},(q+1) \bmod 3,j,t)}\subseteq \sblock_{3\ell + q}\cup\sblock_{3(\ell + \floor{q/2}) + ((q+1) \bmod 3)}$, but there are some exceptions.
Firstly, if $\ell = k-1$ and $q=2$ the set of eligible machines is $\set{\sortmach(k-1,2,j,t),\clausemach(\kappa(j,t)}\subseteq \sblock_{3k -1} \cup \cblock$.
Secondly, for each $\ell \in[k]$ the jobs $\sortjob(\ell,q,j^{>}_\ell,t^{>}_\ell)$ with $q\in [2]$ have size $2$.
Lastly, $\sortjob(\ell,0,j^{>}_\ell,t^{>}_\ell)$ is eligible on $\set{\sortmach(\ell,0,j^{>}_\ell,t^{>}_\ell),\sortmach(\ell,0,j^{<}_\ell,t^{<}_\ell),\sortmach(\ell,1,j^{>}_\ell,t^{>}_\ell)} \subseteq\sblock_{3\ell}\cup\sblock_{3\ell + 1}$ 
and $\sortjob(\ell,0,j^{<}_\ell,t^{<}_\ell)$ on $\set{\sortmach(\ell,0,j^{<}_\ell,t^{<}_\ell),\sortmach(\ell,1,j^{<}_\ell,t^{<}_\ell),\sortmach(\ell,1,j^{>}_\ell,t^{>}_\ell)} \subseteq\sblock_{3\ell}\cup\sblock_{3\ell + 1}$.

\item Each sorting machine $\sortmach(\ell,q,j,t)$ for $\ell \in[k]$, $q\in [3]$, $j\in[n]$, and $t\in[4]$ with $(\ell,q,j,t) \notin \set{(\ell,0,j^{>}_\ell,t^{>}_\ell), (\ell,1,j^{>}_\ell,t^{>}_\ell), (\ell,2,j^{>}_\ell,t^{>}_\ell), (\ell,2,j^{<}_\ell,t^{<}_\ell) } $ has a private load of $1$.

\item Amplifier bridge jobs $\ampbridgejob(\ell,q)$ for each $\ell \in[k]$ and $q\in[2]$ each of size $1$ and and eligible on $\set{\ampmach(\ell,q),\ampmach(\ell,q+1)} \subseteq \sblock_{3\ell + q}\cup\sblock_{3\ell + q + 1}$. 

\item Amplifier shift jobs $\ampshiftjob(\ell,q)$ for each $\ell \in[k]$ and $q\in[3]$ each of size $1$.
The job $\ampshiftjob(\ell,0)$ is eligible on $\set{\ampmach(\ell,0),\sortmach(\ell,0,j^{>}_\ell,t^{>}_\ell)} \subseteq \sblock_{3\ell}$, the second job $\ampshiftjob(\ell,1)$ on $\set{\ampmach(\ell,2),\sortmach(\ell,2,j^{<}_\ell,t^{<}_\ell)} \subseteq \sblock_{3\ell + 1}$, and the last one $\ampshiftjob(\ell,2)$ on $\set{\sortmach(\ell,2,j^{<}_\ell,t^{<}_\ell),\sortmach(\ell,2,j^{>}_\ell,t^{>}_\ell)} \subseteq \sblock_{3\ell + 2}$.

\item Each amplifier machine $\ampmach(\ell,q)$ for $\ell \in[k]$ and $q\in [3]$ has a private load of $1$.

\item One clause job $\clausejob(i,s)$ for each $i\in[2m]$ and $s\in[3]$. 
The job has size $1$, if $s=1$ or $s=2$ and $C_i$ is a 2-in-3-clause, or size $2$, if $s=3$ or $s=2$ and $C_i$ is a 1-in-3-clause.
It is eligible on $\sett{\clausemach(i,s')}{s'\in[3]} \subseteq \cblock$.
\end{itemize}

\subparagraph{Analysis.}

In some sense the present reduction is very similar to the first one that was introduced.
In particular, the truth assignment decisions are made in the truth assignment gadget, and these decisions are transmitted to the clause gadget where they are verified.
Both gadgets are essentially unchanged.
This also holds for the first step of the transmission process:
Either the variable jobs corresponding to the two positive occurrences of a variable are send out of the truth assignment gadget, or the two jobs corresponding to the two negative occurrences.
In each block of the sorting gadget there are job copies for each job occurrence that essentially forward the transmission.
There is, however, in each triple of sorting blocks, a pair of jobs with some ambiguity that we address in two of the following three claims.
The reasoning behind the sorting gadget, which may seem somewhat bloated at this point, will get clear in the next paragraph.
As usual, first note:
\begin{claim}
The overall job size $\sum_{j\in\jobs} p(j)$ is equal to $2|\machs|$.
\end{claim}
\begin{claimproof}
There are $2n + 12kn + 3k + 6m = (12k + 6)n + 3k$ machines (since $6m = 4n$).
On the other hand, the truth assignment jobs have overall size $2n$, the variable jobs $4n$, the sorting jobs $12kn + 2k$, the private loads of the sorting machines $12nk - 4k$, the amplifier bridge jobs $2k$, the amplifier shift jobs $3k$, the private loads of the amplifier machines $3k$, and the clause jobs $ 3\cdot 3m = 6n$.
Hence, we have $\sum_{j\in\jobs} p(j) = (24k + 12)n + 6k = 2|\machs|$.
\end{claimproof}
In the next two claims, we consider the sorting gadget which is also considered in \cref{fig:rank3_sorting_gadget}.
\begin{claim}
Let $\ell \in[k]$.
In any $2$-schedule, the sorting job $\sortjob(\ell,0,j^{>}_\ell,t^{>}_\ell)$ may only be scheduled on a machine from $\set{\sortmach(\ell,0,j^{>}_\ell,t^{>}_\ell),\sortmach(\ell,1,j^{>}_\ell,t^{>}_\ell)}$ and $\sortjob(\ell,0,j^{<}_\ell,t^{<}_\ell)$ only on one from $\set{\sortmach(\ell,0,j^{<}_\ell,t^{<}_\ell),\sortmach(\ell,1,j^{<}_\ell,t^{<}_\ell)}$.
\end{claim}
\begin{claimproof}
The only other machine $\sortjob(\ell,0,j^{>}_\ell,t^{>}_\ell)$ is eligible on is $\sortmach(\ell,0,j^{<}_\ell,t^{<}_\ell)$.
However, $\sortmach(\ell,0,j^{<}_\ell,t^{<}_\ell)$ has a private load of $1$ and $\sortjob(\ell,0,j^{>}_\ell,t^{>}_\ell)$ has a size of $2$. 
Hence, $\sortjob(\ell,0,j^{>}_\ell,t^{>}_\ell)$ cannot be scheduled on $\sortmach(\ell,0,j^{<}_\ell,t^{<}_\ell)$ in any $2$-schedule.
Similarly, the only other machine $\sortjob(\ell,0,j^{<}_\ell,t^{<}_\ell)$ is eligible on is $\sortmach(\ell,1,j^{>}_\ell,t^{>}_\ell)$.
But $\sortjob(\ell,0,j^{<}_\ell,t^{<}_\ell)$ is the only job with size $1$ eligible on this machine and the other eligible jobs have size $2$.
Hence, $\sortjob(\ell,0,j^{<}_\ell,t^{<}_\ell)$ cannot be scheduled on $\sortmach(\ell,1,j^{>}_\ell,t^{>}_\ell)$.
\end{claimproof}
\begin{claim}
For each $j\in[n]$ and $t\in[4]$ the following is true in any $2$-schedule:
If $\variablejob(j,t)$ is scheduled on $\sortmach(0,0,j,t)$, then $\sortjob(\ell,q,j,t)$ is scheduled on $\sortmach(\ell + \floor{q/2}, (q+1)\bmod 3,j,t)$ for each $\ell\in[k]$ and $q\in[3]$ with $(\ell,q)\neq(k-1,2)$ and $\sortjob(k-1,2,j,t)$ is scheduled on $\clausemach(\kappa(j,t))$.
If, on the other hand, $\variablejob(j,t)$ is scheduled on $\truthmach(j,\floor{t/2})$, then $\sortjob(\ell,q,j,t)$ is scheduled on $\sortmach(\ell, q,j,t)$ for each $\ell\in[k]$ and $q\in[3]$.
\end{claim}
\begin{claimproof}
Each of the mentioned jobs may only be scheduled on one of two machines in any $2$-schedule and the two machines are positioned in two directly succeeding blocks of machines. 
Hence, we call the machine in the earlier block the left and the one in the later block the right machine of the respective job.
Now, $\variablejob(j,t)$ may be scheduled on its left or right machine, i.e., $\truthmach(j,\floor{t/2})$ or $\sortmach(0,0,j,t)$.
If $(j,t) \notin \set{(j^{>}_0,t^{>}_0),(j^{<}_0,t^{<}_0)}$, then $\sortmach(0,0,j,t)$ has a private load of $1$ and $\variablejob(j,t)$ and $\sortjob(0,0,j,t)$ are the only eligible jobs on this machine and both have size $1$.
Note that $\sortmach(0,0,j,t)$ is the left machine of $\sortjob(0,0,j,t)$.
Therefore, if $\variablejob(j,t)$ is scheduled on its left or right machine, the same has to hold for $\sortjob(0,0,j,t)$.
The same argument can be repeated for $\sortjob(\ell,q,j,t)$ for lexicographically increasing $(\ell,q)$.
Hence, we have to consider the case $(j,t) \in \set{(j^{>}_\ell,t^{>}_\ell),(j^{<}_\ell,t^{<}_\ell)}$ (see \cref{fig:rank3_sorting_gadget}).
We again only consider the case $\ell = 0$ because the case $\ell>0$ works analogously.
\begin{figure}
\centering
\begin{tikzpicture}
\pgfmathsetmacro{\MachW}{0.8}
\pgfmathsetmacro{\MachH}{0.25}
\pgfmathsetmacro{\MachGap}{4.1*\MachW}
\pgfmathsetmacro{\JobGap}{0.07}
\pgfmathsetmacro{\JobH}{\MachW - 2*\JobGap}
\pgfmathsetmacro{\JobW}{\MachW - 2*\JobGap}
\pgfmathsetmacro{\JobMult}{0.8}
\pgfmathsetmacro{\RJobH}{1.12}
\pgfmathsetmacro{\LJobH}{2.35}
\pgfmathsetmacro{\LJobHS}{\LJobH + \JobMult*\JobH}


\foreach \x/\y/\l in {	0/0/$\mathtt{AM}_0$,1/0/$\mathtt{SM}_0^>$,2/0/$\mathtt{SM}_0^<$,
						0/1/$\mathtt{AM}_1$,1/1/$\mathtt{SM}_1^<$,2/1/$\mathtt{SM}_1^>$,
						0/2/$\mathtt{AM}_2$,1/2/$\mathtt{SM}_2^<$,2/2/$\mathtt{SM}_2^>$}
{
\draw[thick] (\x * \MachW + \y*\MachGap,\MachH) -- (\x * \MachW + \y*\MachGap,0) node[xshift = 0.5*\MachW cm, yshift = -0.25*\MachW cm] {\footnotesize  \l} --  (\x * \MachW + \MachW + \y*\MachGap,0) -- (\x * \MachW + \MachW +\y*\MachGap,\MachH);
}



\foreach \x/\y in {	0/0, 2/0, 0/1, 1/1, 0/2}
{
\draw[thick, pattern = north west lines] (\JobGap + \x*\MachW + \y*\MachGap,\JobGap) rectangle (\JobGap + \JobW + \x*\MachW + \y*\MachGap,\JobGap + \JobW);
}


\foreach \x/\y/\l in {	0/0/$\mathtt{AS}_0$, 0/2/$\mathtt{AS}_1$, 1/2/$\mathtt{AS}_2$}
{
\draw[thick, color = mblue] (0.75*\MachW + \x*\MachW + \y*\MachGap,-0.45*\MachW ) -- ++(0.25*\MachW,-0.5*\MachW) -- ++(0.25*\MachW,0.5*\MachW);

\draw[thick, color = mblue] (1*\MachW - 0.5*\JobMult*\JobH + \x*\MachW + \y*\MachGap,-0.95*\MachW - \JobGap ) rectangle ++(\JobMult*\JobH, -\JobMult*\JobH) node[midway] {\scriptsize \l};
}


\foreach \y/\l in {0/$\mathtt{AB}_0$, 1/$\mathtt{AB}_1$}
{
\draw[thick, color = mdarkorange] (0.25*\MachW + \y*0.5*\MachW + \y*\MachGap,-0.45*\MachW ) -- ++(0,-1.4*\MachW) -- ++(\MachGap - \y*0.5*\MachW,0) -- ++(0,1.4*\MachW);

\draw[thick, color = mdarkorange] (\y*0.25*\MachW + \y*\MachGap + 0.5*\MachGap,-1.85*\MachW + \JobGap ) rectangle ++(\JobMult*\JobH, \JobMult*\JobH) node[midway] {\scriptsize \l};
}



\draw[thick, dotted, color = mdarkblue]  (-0.55,\RJobH) -- (-0.3,\RJobH);
\draw[thick, color = mdarkblue]  (-0.3,\RJobH) -- (0.25*\MachW + 1*\MachW,\RJobH) -- (0.25*\MachW + 1*\MachW,\MachW);

\draw[thick, color = mdarkblue] (0.2*\MachW,\RJobH + \JobGap ) rectangle ++(\JobMult*\JobH, \JobMult*\JobH) node[midway] {\scriptsize $\mathtt{SJ}_0^>$};


\draw[thick, color = mdarkblue]  (0.75*\MachW + 1*\MachW,\MachW) --  (0.75*\MachW + 1*\MachW,\RJobH) -- (0.25*\MachW + 2*\MachW + 1*\MachGap,\RJobH) -- (0.25*\MachW + 2*\MachW + 1*\MachGap,\MachW);
\draw[thick, color = mdarkblue] (0.5*\MachW + 2*\MachW + 0*\MachGap,\RJobH) -- (0.5*\MachW + 2*\MachW + 0*\MachGap,\MachW);

\draw[thick, color = mdarkblue] (0.9*\MachGap,\RJobH + \JobGap ) rectangle ++(\JobMult*\JobH, 2*\JobMult*\JobH) node[midway] {\scriptsize $\mathtt{SJ}^>_1$};


\draw[thick, color = mdarkblue]  (0.75*\MachW + 2*\MachW + 1*\MachGap,\MachW) --  (0.75*\MachW + 2*\MachW + 1*\MachGap,\RJobH) -- (0.25*\MachW + 2*\MachW + 2*\MachGap,\RJobH) -- (0.25*\MachW + 2*\MachW + 2*\MachGap,\MachW);

\draw[thick, color = mdarkblue] (2.5*\MachW + 1.5*\MachGap -0.5*\JobMult*\JobH,\RJobH + \JobGap ) rectangle ++(\JobMult*\JobH, 2*\JobMult*\JobH) node[midway] {\scriptsize $\mathtt{SJ}_2^>$};


\draw[thick, color = mdarkblue] (0.75*\MachW + 2*\MachW + 2*\MachGap,\MachW) -- (0.75*\MachW + 2*\MachW + 2*\MachGap,\RJobH) -- ( 0.26 + 3*\MachW + 2*\MachGap,\RJobH);  
\draw[thick, dotted, color = mdarkblue]  ( 0.26 + 3*\MachW + 2*\MachGap,\RJobH) -- ++(0.33 ,0); 

\draw[thick, color = mdarkblue] (3.1*\MachW + 2*\MachGap,\RJobH + \JobGap) rectangle ++(\JobMult*\JobH, \JobMult*\JobH) node[midway] {\scriptsize $\mathtt{SJ}_3^>$};



\draw[thick, dotted, color = mred]  (-0.55,\LJobH) -- (-0.3,\LJobH);
\draw[thick, color = mred]  (-0.3,\LJobH) -- (0.25*\MachW + 2*\MachW,\LJobH) -- (0.25*\MachW + 2*\MachW,\MachW);

\draw[thick, color = mred] (0.7*\MachW,\LJobH + \JobGap ) rectangle ++(\JobMult*\JobH, \JobMult*\JobH) node[midway] {\scriptsize $\mathtt{SJ}_0^<$};


\draw[thick, color = mred]  (0.75*\MachW + 2*\MachW,\MachW) --  (0.75*\MachW + 2*\MachW,\LJobH) -- (0.5*\MachW + 2*\MachW + 1*\MachGap,\LJobH) -- (0.5*\MachW + 2*\MachW + 1*\MachGap,\MachW);
\draw[thick, color = mred] (0.25*\MachW + 1*\MachW + 1*\MachGap,\LJobH) -- (0.25*\MachW + 1*\MachW + 1*\MachGap,\MachW);

\draw[thick, color = mred] (1.05*\MachGap,\LJobH + \JobGap ) rectangle ++(\JobMult*\JobH, \JobMult*\JobH) node[midway] {\scriptsize $\mathtt{SJ}_1^<$};


\draw[thick, color = mred]  (0.75*\MachW + 1*\MachW + 1*\MachGap,\MachW) --  (0.75*\MachW + 1*\MachW + 1*\MachGap,\LJobHS) -- (0.25*\MachW + 1*\MachW + 2*\MachGap,\LJobHS) -- (0.25*\MachW + 1*\MachW + 2*\MachGap,\MachW);

\draw[thick, color = mred] (1.5*\MachW + 1.5*\MachGap -0.5*\JobMult*\JobH,\LJobHS - \JobGap ) rectangle ++(\JobMult*\JobH, -1*\JobMult*\JobH) node[midway] {\scriptsize $\mathtt{SJ}_2^<$};


\draw[thick, color = mred] (0.75*\MachW + 1*\MachW + 2*\MachGap,\MachW) -- (0.75*\MachW + 1*\MachW + 2*\MachGap,\LJobH) -- ( 0.26 + 3*\MachW + 2*\MachGap,\LJobH);  
\draw[thick, dotted, color = mred]  ( 0.26 + 3*\MachW + 2*\MachGap,\LJobH) -- ++(0.33 ,0); 

\draw[thick, color = mred] (2.6*\MachW + 2*\MachGap,\LJobH + \JobGap) rectangle ++(\JobMult*\JobH, \JobMult*\JobH) node[midway] {\scriptsize $\mathtt{SJ}_3^<$};

\end{tikzpicture}
\caption{
The core part of the sorting gadget for the \lrs{3} reduction for a fixed $\ell\in[k]$ involving three machines in each of the three blocks $\sblock_{3\ell}$, $\sblock_{3\ell + 1}$, and $\sblock_{3\ell + 2}$.
Jobs are represented by rectangles with heights corresponding to size, hatched rectangles representing private loads, and connecting lines at the remaining rectangles (which are slightly scaled down) corresponding to eligibilities.
The following notation is used:
$\mathtt{AM}_q = \ampmach(\ell,q)$, $\mathtt{SM}_q^\circ =  \sortmach(\ell,q,j_\ell^\circ,t_\ell^\circ)$, $\mathtt{AS}_q =  \ampshiftjob(\ell,q)$, $\mathtt{AB}_q =  \ampbridgejob(\ell, q)$,  $\mathtt{SJ}_q^\circ  = \sortjob(\ell,q-1,j_\ell^\circ,t_\ell^\circ)$ for $q>0$, and $\mathtt{SJ}_0^\circ =  \variablejob(j_\ell^\circ,t_\ell^\circ)$ if $\ell = 0$ and $\mathtt{SJ}_0^\circ =  \sortjob(\ell - 1,2,j_\ell^\circ,t_\ell^\circ)$ otherwise.
} 
\label{fig:rank3_sorting_gadget}
\end{figure}
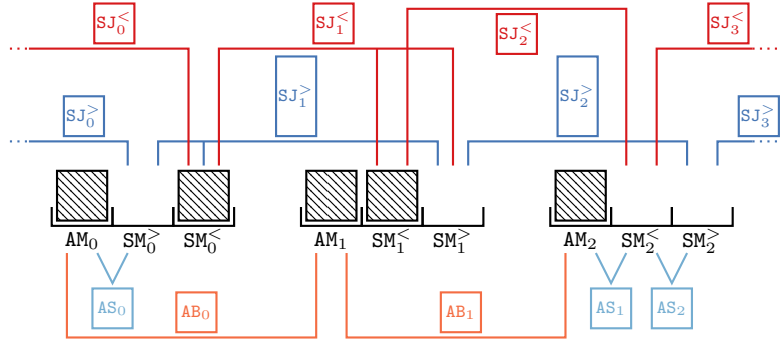
We consider two cases.
First, assume that $\variablejob(j^{>}_0,t^{>}_0)$ is scheduled on its left machine.
The only other eligible jobs on $\sortmach(0,0,j^{>}_0,t^{>}_0)$ are $\sortjob(0,0,j^{>}_0,t^{>}_0)$ and $\ampshiftjob(0,0)$. 
The former has size $2$ and the latter has size $1$.
Hence, $\sortjob(0,0,j^{>}_0,t^{>}_0)$ has to be scheduled on $\sortmach(0,0,j^{>}_0,t^{>}_0)$ and $\ampshiftjob(0,0)$ on its only other eligible machine, namely $\ampmach(0,0)$.
Since, the amplifier machines each have a private load of $1$, this determines the schedule for $\ampmach(0,1)$ and $\ampmach(0,2)$ as well, i.e., $\ampmach(0,1)$ has to receive the amplifier bridge job $\ampbridgejob(0,0)$ and $\ampmach(0,2)$ the job $\ampbridgejob(0,1)$.
Similarly, $\sortmach(0,1,j^{>}_0,t^{>}_0)$ has to receive $\sortjob(0,1,j^{>}_0,t^{>}_0)$ since this is the only other job that can be scheduled on this machine in a $2$-schedule.
The remaining eligible jobs on $\sortmach(0,2,j^{>}_0,t^{>}_0)$ are now $\sortjob(0,2,j^{>}_0,t^{>}_0)$ and $\ampshiftjob(0,2)$ and both have size $1$.
Hence, both have to be scheduled on this machine. 
The last amplifier shift job $\ampshiftjob(0,2)$ cannot be scheduled on $\ampmach(0,2)$ due to the above considerations and is therefore scheduled on $\sortmach(0,2,j^{<}_0,t^{<}_0)$.
Now, taking the schedule of all of the above jobs into account, it is easy to that the placement of $\variablejob(j^{<}_0,t^{<}_0)$ on its left or right machine implies that $\sortjob(0,q,j^{<}_0,t^{<}_0)$ is placed on its left or right machine, respectively, as well.

Now, we assume that $\variablejob(j^{>}_0,t^{>}_0)$ is scheduled on its right machine.
This implies that $\sortjob(0,0,j^{>}_0,t^{>}_0)$ has to placed on $\sortmach(0,1,j^{>}_0,t^{>}_0)$ and $\ampshiftjob(0,0)$ has to be placed on $\sortmach(0,0,j^{>}_0,t^{>}_0)$.
This directly determines the schedules for the amplifier machines, in particular $\ampmach(0,0)$ has to receive $\ampbridgejob(0,0)$, $\ampmach(0,1)$ the job $\ampbridgejob(0,1)$, and $\ampmach(0,2)$ the amplifier shift job $\ampshiftjob(0,1)$.
Furthermore, $\sortjob(0,1,j^{>}_0,t^{>}_0)$ has to be placed on $\sortmach(0,2,j^{>}_0,t^{>}_0)$ and $\sortjob(0,2,j^{>}_0,t^{>}_0)$ on $\sortmach(1,0,j^{>}_0,t^{>}_0)$.
Since $\sortjob(0,1,j^{>}_0,t^{>}_0)$ has size $2$, the last amplifier shift job $\ampshiftjob(0,2)$ has to be placed on $\sortmach(0,2,j^{<}_0,t^{<}_0)$.
Again, taking the schedule of all of the above jobs into account, it is easy to see that the placement of $\variablejob(j^{<}_0,t^{<}_0)$ on its left or right machine implies that $\sortjob(0,q,j^{<}_0,t^{<}_0)$ is placed on its left or right machine, respectively, as well.
\end{claimproof}
The above claims and the same arguments used for \cref{lem:basic_reduction} yield:
\begin{lemma}\label{lem:rank3_RA_instance}
There is satisfying assignment for the $\tailoredsat$ instance, if and only if there is there is a $2$-schedule for the constructed restricted assignment instance.
\end{lemma}

\subparagraph{Speeds and Sizes.}

In this paragraph we show that the sizes and eligibilities can be approximated with arbitrary precision by an \lrs{3} instance.
Before formally stating how this is achieved, we first discuss the idea of the approach.
Note that in our definition of size and speed vectors a higher numerical value in a speed vector corresponds to \emph{bigger} processing times.
Since this is counterintuitive, we will, in the following brief discussion, call the multiplicative inverse values of the numerical speed values the \emph{speeds} of the machines (all such values will be larger than zero). 

When considering a single block of machines the speeds in the first dimension are decreasing and the speeds in the other two dimensions are increasing with respect to the current ordering of the pairs $(j,t)\in [n] \times [4]$.
Furthermore, when comparing succeeding blocks of machines, the speeds in the first dimensions stay essentially the same, the ones in the second are smaller than in any predecessor block and bigger than in any successor block, and the speeds in last dimension are bigger than in any predecessor block and smaller than in any successor block.
Using this structure, we can realize jobs that are eligible on exactly two machines in two succeeding blocks as well as jobs eligible on few machines belonging to the same block. 
The idea is visualized in \cref{fig:rank3_speeds_idea} and made concrete in the following.
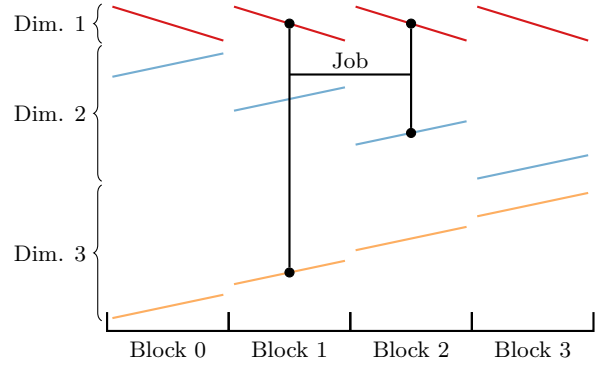
\begin{figure}
\centering
\begin{tikzpicture}
\pgfmathsetmacro{\W}{1.6}
\pgfmathsetmacro{\H}{0.45}
\pgfmathsetmacro{\bH}{0.25}
\pgfmathsetmacro{\Gap}{0.07}
\pgfmathsetmacro{\FSH}{4.3}
\pgfmathsetmacro{\SSH}{3.3}
\pgfmathsetmacro{\TSH}{0.1}

\foreach \x in {0,1,2,3}
{
\draw[thick] (\x * \W ,\bH) -- (\x * \W,0) node[xshift = 0.5*\W cm, yshift = -0.15*\W cm] {\footnotesize  Block \x}  --  (\x * \W + \W ,0) -- (\x * \W + \W ,\bH); 
}


\foreach \x in {0,1,2,3}
{
\draw[thick, color = mred] (\x * \W + \Gap ,\FSH) -- (\x * \W + \W - \Gap ,\FSH - \H);
}


\draw[decorate,decoration={brace,amplitude=3pt,raise=2pt, mirror},yshift=0cm] (0,\FSH + 0.5*\Gap) -- (0,\FSH - \H -0.5*\Gap) node [midway,xshift=-0.75cm]{\footnotesize Dim. 1};


\foreach \x in {0,1,2,3}
{
\draw[thick, color = mblue] (\x * \W + \Gap ,\SSH - \x*\H + \Gap) -- (\x * \W + \W - \Gap ,\SSH -\x*\H + \H -\Gap);
}

\draw[decorate,decoration={brace,amplitude=3pt,raise=2pt, mirror},yshift=0cm] (0,\SSH + 0.5*\Gap + \H) -- (0,\SSH + 0.5*\Gap - 3*\H) node [midway,xshift=-0.75cm]{\footnotesize Dim. 2};


\foreach \x in {0,1,2,3}
{
\draw[thick, color = morange] (\x * \W + \Gap ,\TSH + \x*\H + \Gap) -- (\x * \W + \W - \Gap ,\TSH +\x*\H + \H -\Gap);
}

\draw[decorate,decoration={brace,amplitude=3pt,raise=2pt},yshift=0cm] (0,\TSH + 0.5*\Gap) -- (0,\TSH + 0.5*\Gap + 4*\H) node [midway,xshift=-0.75cm]{\footnotesize Dim. 3};


\node[draw,circle,fill=black,inner sep = 1.2] (dim1r) at (2.5*\W ,\FSH - 0.5*\H) {};
\node[draw,circle,fill=black,inner sep = 1.2] (dim1l) at (1.5*\W ,\FSH - 0.5*\H) {};
\node[draw,circle,fill=black,inner sep = 1.2] (dim2r) at (2.5*\W ,\SSH -2*\H + 0.5*\H) {};
\node[draw,circle,fill=black,inner sep = 1.2] (dim3l) at (1.5*\W ,\TSH +1*\H + 0.5*\H) {};

\draw[thick] (dim1r) -- (dim2r);
\draw[thick] (dim1l) -- (dim3l);
\draw[thick] ($(dim1r) + (0,-1.5*\H)$) -- ++(-\W,0) node[midway, yshift = 0.12*\W cm] {\footnotesize Job};

\end{tikzpicture}
\caption{
The basic idea for the speeds of the machines in the \lrs{3} reduction.
In the picture, we do not depict the individual machines but assume that a machine corresponds to a horizontal position in the block and that there is some minimum speed difference in each dimension between succeeding machines.
If we draw a vertical line starting from this point, the speeds of the machines correspond to the heights of the three interception points.
Now, we could the sizes of a job such that its processing time due to each dimension becomes reasonably small starting from a certain speed value, and such a choice is depicted in the picture by the connected dots labeled \enquote{job}. 
For this job, it can be guaranteed that it cannot be scheduled on machines from block 0 or 3 due to the third and second dimension, respectively, and only on one machine from each of the blocks 1 and 2, due to the first and third or first and second dimension, respectively.
} 
\label{fig:rank3_speeds_idea}
\end{figure}

Let $\delta \in (0,1]$, $K\geq 1$, $\eps = \delta / 2$, $N = K/\eps$, and $C = 16n $.
For each $\ell\in[k+1]$, $j\in[n]$, and $t\in[4]$, let $\iota(\ell, j,t)$ be the index of the pair $(j,t)$ in the sequence $\varphi_{\ell}$ (with indexing being started at 0).
Furthermore, let $\iota^*_\ell = \iota(\ell,j^{>}_\ell,t^{>}_\ell)$ for each $\ell\in[k]$.
\begin{observation}\label{obs:rank3_indices}
We have:
\begin{itemize}
\item $\iota(0,j,t) = 4j + t$ for each $j\in[n]$ and $t\in[4]$.
\item $\iota(k,\kappa^{-1}(i,s)) = 3i + s$ for each $i\in[2m]$ and $s\in[3]$.
\item $\iota(\ell,j,t) = \iota(\ell + 1,j,t)$ for each $\ell\in[k]$ and $(j,t)\in[n]\times[4]\setminus\set{(j^{>}_\ell,t^{>}_\ell),(j^{<}_\ell,t^{<}_\ell)}$.
\item $\iota(\ell,j^{>}_\ell,t^{>}_\ell) = \iota^*_\ell$ and $\iota(\ell + 1,j^{>}_\ell,t^{>}_\ell) = \iota^*_\ell + 1$ for each $\ell\in[k]$.
\item $\iota(\ell,j^{<}_\ell,t^{<}_\ell) = \iota^*_\ell + 1$ and $\iota(\ell+1,j^{<}_\ell,t^{<}_\ell) = \iota^*_\ell$ for each $\ell\in[k]$.
\end{itemize}
\end{observation}
The speed vectors of the machines are specified in \cref{table:rank3_speed_vectors} and the size vectors of the jobs  in \cref{table:rank3_size_vectors}.
\begin{table}
\centering
\caption{The speed vectors of the machines in the \lrs{3} reduction.}
\begin{tabular}{llll}
\toprule
Machine & Dim. $1$ & Dim. $2$ & Dim. $3$ \\ \midrule

$\truthmach(j,q)$ & 
$(\frac{1}{N})^{-2(4j + 2q)}$ & 
$(\frac{1}{N})^{C + 2j+q}$ & 
$(\frac{1}{N})^{- C + 2j + q}$  \\

$\sortmach(\ell,q,j,t)$ & 
$(\frac{1}{N})^{-2\iota(\ell + \ceil{q/2},j,t)}$ &
$(\frac{1}{N})^{-(3\ell+ q) \cdot C + 2\iota(\ell + \ceil{q/2},j,t)}$ &
$(\frac{1}{N})^{(3\ell+q)\cdot C + 2\iota(\ell + \ceil{q/2},j,t)}$  \\


%

$\ampmach(\ell,q)$ & 
$(\frac{1}{N})^{-2\iota^*_\ell +1}$ & 
$(\frac{1}{N})^{-(3\ell +q)\cdot C + 2\iota^*_\ell - 1}$ & 
$(\frac{1}{N})^{(3\ell+q)\cdot C + 2\iota^*_\ell - 1}$  \\

$\clausemach(i,s)$ & 
$(\frac{1}{N})^{-2(3i + s)}$ & 
$(\frac{1}{N})^{-3k\cdot C + 2(3i + s)}$ & 
$(\frac{1}{N})^{3k\cdot C + i}$  \\
\bottomrule
\end{tabular}
\label{table:rank3_speed_vectors}
\end{table}
\begin{table}
\centering
\caption{The sizes of the jobs in the \lrs{3} reduction. We set $\phi(i,0) = 1$, $\phi(i,2) = 2$, and $\phi(i,1) = 3-k$ if $C_i$ is a $k$-in-3-clause.
The row regarding $\sortjob(\ell,q,j,t)$ considers all tuples $(\ell,q,j,t)$ except the ones considered directly thereafter.}
\begin{tabular}{llll}
\toprule
Job & Dim. $1$ & Dim. $2$ & Dim. $3$ \\ \midrule

$\truthjob(j)$ & 
$2\cdot N^{-2(4j + 2)}$ & 
$2\cdot N^{C + 2j}$ & 
0  \\

$\variablejob(j,t)$ & 
$\eps \cdot N^{-2(4j + t)}$ & 
$N^{2(4j + t)}$ & 
$N^{-C + 2j + \floor{t/2}}$  \\

$\sortjob(\ell,q,j,t)$ & 
$\eps \cdot N^{-2\iota(\ell + 1,j,t)}$ & 
$ N^{-(3\ell + q + 1)\cdot C + 2\iota(\ell + 1,j,t)}$ & 
$N^{(3\ell + q)\cdot C + 2\iota(\ell + 1,j,t)}$  \\

$\sortjob(\ell,0,j^{>}_\ell,t^{>}_\ell)$ & 
$2\cdot N^{-2(\iota^*_\ell + 1)}$ & 
$\eps\cdot N^{-(3\ell + 1)\cdot C + 2(\iota^*_\ell + 1)}$ & 
$2\cdot N^{(3\ell)\cdot C + 2\iota^*_\ell }$  \\

$\sortjob(\ell,0,j^{<}_\ell,t^{<}_\ell)$ & 
$N^{-2(\iota^*_\ell + 1)}$ & 
$N^{-(3\ell +1)\cdot C + 2\iota^*_\ell}$ & 
$\eps\cdot N^{(3\ell)\cdot C + 2(\iota^*_\ell + 1)}$  \\

$\sortjob(\ell,1,j^{>}_\ell,t^{>}_\ell)$ & 
$\eps \cdot N^{-2(\iota^*_\ell + 1)}$ & 
$2\cdot N^{-(3\ell + 2)\cdot C + 2(\iota^*_\ell + 1)}$ & 
$2\cdot N^{(3\ell + 1)\cdot C + 2(\iota^*_\ell + 1)}$  \\

$\ampbridgejob(\ell,q)$ & 
$\eps \cdot N^{-2\iota^*_\ell +1}$ & 
$N^{-(3\ell + q + 1)\cdot C + 2\iota^*_\ell - 1}$ & 
$N^{(3\ell+q)\cdot C + 2\iota^*_\ell - 1}$  \\

$\ampshiftjob(\ell,0)$ & 
$N^{-2\iota^*_\ell}$ & 
$\eps \cdot N^{-(3\ell)\cdot C + 2\iota^*_\ell - 1}$ & 
$N^{(3\ell)\cdot C + 2\iota^*_\ell - 1}$  \\

$\ampshiftjob(\ell,1)$ & 
$N^{-2\iota^*_\ell}$ & 
$\eps \cdot N^{-(3\ell +2)\cdot C + 2\iota^*_\ell - 1}$ & 
$N^{(3\ell+2)\cdot C + 2\iota^*_\ell - 1}$  \\

$\ampshiftjob(\ell,2)$ & 
$N^{-2(\iota^*_\ell + 1)}$ & 
$\eps \cdot N^{-(3\ell +2)\cdot C + 2\iota^*_\ell}$ & 
$N^{(3\ell+2)\cdot C + 2\iota^*_\ell}$  \\

$\sortmach(\ell,q,j,t)$ &
$N^{-2\iota(\ell + \ceil{q/2},j,t)}$ &
$\eps \cdot N^{-(3\ell +q)\cdot C + 2\iota(\ell + \ceil{q/2},j,t)}$ &
$\eps \cdot N^{(3\ell+q)\cdot C + 2\iota(\ell + \ceil{q/2},j,t)}$  \\

$\ampmach(\ell,q)$ &
$N^{-2\iota^*_\ell +1}$ & 
$\eps \cdot N^{-(3\ell +q)\cdot C + 2\iota^*_\ell - 1}$ & 
$\eps \cdot N^{(3\ell+q)\cdot C + 2\iota^*_\ell - 1}$  \\

$\clausejob(i,s)$ & 
$\eps \cdot N^{-2(3i + 2)}$ & 
$0$ & 
$\phi(i,s) \cdot N^{3k\cdot C + i}$  \\

\bottomrule
\end{tabular}
\label{table:rank3_size_vectors}
\end{table}
\begin{lemma}\label{lem:lrs3_speeds_and_sizes}
Let $I$ be the described restricted assignment instance and $I'$ the \lrs{3} instance specified by the size and speed vectors.
Let furthermore, $p_j$ be the size of a job $j$ in instance $I$, $\emachs(j)$ the corresponding set of eligible machines, and $p'_{ij}$ the size of $j$ on machine $i$ in instance $I'$.
We have $p_{j} \leq p_{ij} \leq p_j + \delta $ if $i\in\emachs(j)$ and $p_{ij} > K $ otherwise.
\end{lemma}
The proof of this lemma is unfortunately rather tedious since essentially all pairs of job and machine types including the special cases have to be considered. 
However, since $\delta$ and $K$ can be chosen arbitrarily, \cref{lem:lrs3_speeds_and_sizes} and \cref{lem:rank3_RA_instance} directly yield the result of the current section:
\begin{theorem}
There is no better than $\frac{3}{2}$-approximation for \lrs{3}, unless P=NP.
\end{theorem}
\begin{proof}[Proof of \cref{lem:lrs3_speeds_and_sizes}.]
First note that since the second speed dimension is increasing from block to block and the third is decreasing, the jobs have sizes of at least $N>K$ on most blocks.
Moreover, the blocks of machines that allow sizes smaller than $N$ for any jobs correspond to the blocks containing the eligible machines in $I$ for each job.
In particular, these blocks are $\tblock$ for the truth assignment jobs,
$\tblock$ and $\sblock_0$ for the variable jobs,
$\sblock_{3\ell + q}$ and $\sblock_{3(\ell + \floor{q/2}) + ((q+1) \bmod 3)}$ for sorting jobs $\sortjob(\ell,q,j,t)$ with $(\ell,q) \neq (k-1,2)$,
$\sblock_{3k - 1} $ and $\cblock$ for sorting jobs $\sortjob(k-1,2,j,t)$, 
$\sblock_{3\ell + q}$ and $\sblock_{3\ell + q + 1}$ for amplifier bridge jobs $\ampbridgejob(\ell,q)$,
$\sblock_{3\ell}$ or $\sblock_{3\ell + 2}$ for amplifier shift jobs $\ampshiftjob(\ell,q)$ with $q=0$ or $q\in\set{1,2}$, respectively,
$\cblock$ for the clause jobs,
and the block containing its respective machine for any private load. 
 
In the remainder of the proof, we consider the job types one after another and their processing times on machines belonging to their respective blocks.

\proofsubparagraph*{Truth assignment jobs.} The size of $\truthjob(j)$ on $\truthmach(j,0)$ and $\truthmach(j,1)$ is $2N^{-4} + 2$ and $2 + 2N^{-1}$, respectively.
Furthermore, $\truthjob(j)$ has size at least $N$ on $\truthmach(j',q)$ due to the second dimension if $j'<j$ and due to the first dimension if $j'>j$. 

\proofsubparagraph*{Variable jobs.} The size of $\variablejob(j,t)$ on $\truthmach(j,q)$ is $\eps \cdot N^{4q - 2t} + N^{- C + 6j + 2t -q} + N^{\floor{t/2} - q}$.
Hence, it is at least $1$ and at most $\eps + N^{-1} + 1 \leq 1 + \delta$ if $q=0$ and $t\in\set{0,1}$, 
at least $\eps \cdot N^{-6} + N^{- C } + N \geq K$ if $q=0$ and $t\in\set{2,3}$,
at least $\eps \cdot N^{2} + N^{- C} + N^{-1} \geq K$ if $q=1$ and $t\in\set{0,1}$,
and at least $1$ and at most $\eps + N^{- 1} + 1 \leq 1 + \delta$ if $q=1$ and $t\in\set{2,3}$.
Furthermore, $\variablejob(j,t)$ has size at least $N$ on $\truthmach(j',q)$ due to the third dimension if $j'<j$ and due to the first dimension if $j'>j$. 
Next we consider the size of $\variablejob(j,t)$ on sorting machines from $\sblock_{0}$. 
Since $\iota(0, j,t) = 4j +t$, it is $\eps + 1 + N^{-C - 2j - t + \floor{t/2}}$ on $\sortmach(0,0,j,t)$ and at least $N$ on $\sortmach(0,0,j',t')$ due to the second dimension if $j'<j$ and due to the first dimension if $j'>j$.
Lastly, we have to consider the amplifier machine $\ampmach(0,0)$ of $\sblock_{0}$. 
The size of $\variablejob(j,t)$ on this machine is $\eps \cdot N^{2\iota^*_0 - 1 - 2(4j + t)} + N^{2(4j + t) - 2\iota^*_0 + 1} + N^{-C + 2j + \floor{t/2} - 2\iota^*_0 + 1}  >\eps N$ (either due to the first or the second addend).

\proofsubparagraph*{Sorting jobs.}
We consider the sizes of $\sortjob(\ell,q,j,t)$. 
First, assume that $(\ell,q) \neq (k-1,2)$ and $(\ell,q,j,t)\notin\set{(\ell,0,j^{>}_\ell,t^{>}_\ell),(\ell,0,j^{<}_\ell,t^{<}_\ell),(\ell,1,j^{>}_\ell,t^{>}_\ell)}$.
In this case, $\sortjob(\ell,q,j,t)$ has size
$\eps \cdot N^{2\iota(\ell + \ceil{q/2},j',t')-2\iota(\ell,j,t)}+ N^{-C + 2\iota(\ell,j,t) - 2\iota(\ell +  \ceil{q/2},j',t')} + N^{2\iota(\ell,j,t) - 2\iota(\ell + \ceil{q/2},j',t')}$ on a machine $\sortmach(\ell,q,j',t')$ from the block $\sblock_{3\ell + q}$ (taking \cref{obs:rank3_indices} into account).
Note that $\iota(\ell + 1,j',t') = \iota(\ell ,j',t')$ if $(j',t')\notin\set{(j^{>}_\ell,t^{>}_\ell),(j^{<}_\ell,t^{<}_\ell)}$.
Hence, its size is at least $\eps \cdot N^2 $ if $(j,t) \neq (j',t')$ and exactly $\eps + N^{-C } + 1$ otherwise.
Moreover, on the amplifier machine $\ampmach(\ell,q)$ from block $\sblock_{3\ell + q}$ it has size 
$\eps \cdot N^{2\iota^*_\ell - 1 -2\iota(\ell,j,t)} + N^{- C + 2\iota(\ell,j,t) - 2\iota^*_\ell + 1} + N^{2\iota(\ell,j,t) -2\iota^*_\ell +1 } > \eps N$ (either due to the first or the last addend).
Similarly, the size of $\sortjob(\ell,q,j,t)$ on a machine $\sortmach(\ell + \floor{q/2},(q + 1) \bmod 3,j',t')$ from the second block, i.e., $\sblock_{3(\ell + \floor{q/2}) + ((q+1) \bmod 3)}$, is given by 
$\eps \cdot N^{2\iota(\ell + 1,j',t')-2\iota(\ell,j,t)} + 
N^{2\iota(\ell,j,t) - 2\iota(\ell + 1,j',t')} + 
N^{- C + 2\iota(\ell,j,t) - 2\iota(\ell + 1,j',t')}$.
Hence, its size is at least $\eps \cdot N^2 $ if $(j,t) \neq (j',t')$ and $\eps + 1 + N^{-C }$ otherwise (again taking \cref{obs:rank3_indices} into account).
Regarding the amplifier machine $\ampmach(\ell+ \floor{q/2},(q+1) \bmod 3 )$, the size of $\sortjob(\ell,q,j,t)$ is $\eps \cdot N^{2\iota^*_{\ell + \floor{q/2}} - 1-2\iota(\ell,j,t)} + N^{2\iota(\ell,j,t) - 2\iota^*_{\ell + \floor{q/2}} + 1} +N^{-C + 2\iota(\ell,j,t) - 2\iota^*_{\ell + \floor{q/2}} + 1}$ which is again at least $\eps N$ (either due to the first or the second addend).

Next, we consider $\sortjob(k-1,2,j,t)$.
Regarding machines from the first block, i.e., $\sblock_{3k + 1}$, we can use the same argument as before.
On machine $\clausemach(i,s)$ from block $\cblock$ the job has size 
$\eps \cdot N^{2(3i + s) -2\iota(k,j,t)} + N^{2\iota(k,j,t) - 2(3i + s)} + N^{-C + 2\iota(k,j,t) - 2(3i + s)}$.
This amounts to a size of at least $\eps N^2$ if $\iota(k,j,t) \neq 3i + s$ and a size of $\eps + 1 + N^{-C}$ otherwise.
Since $\iota(k,j,t) = 3i + s$ if and only if $\kappa(j,t) = (i,s)$ (see \cref{obs:rank3_indices}), the analysis of this case is therefore completed.

There are three remaining special cases regarding sorting jobs and we start with $\sortjob(\ell,0,j^{>}_\ell,t^{>}_\ell)$.
On a machine $\sortmach(\ell,0,j',t')$ of its first block $\sblock_{3\ell}$ it has a size of 
$2\cdot N^{2\iota(\ell,j',t') - 2(\iota^*_\ell + 1)} + \eps\cdot N^{-C + 2(\iota^*_\ell + 1) - 2\iota(\ell,j',t')} + 2\cdot N^{ 2\iota^*_\ell - 2\iota(\ell,j',t')}$.
Hence, its size is at least $2N^2$ if $\iota(\ell,j',t') \notin \set{\iota^*_\ell, \iota^*_\ell + 1}$.
If $\iota(\ell,j',t') = \iota^*_\ell$, we have a size of $2\cdot N^{-2} + \eps\cdot N^{-C + 2} + 2$,
and if $\iota(\ell,j',t') = \iota^*_\ell + 1$, we have a size of $2 + \eps\cdot N^{-C} + 2\cdot N^{ -2}$.
This is the correct behavior since $\iota(\ell,j',t') = \iota^*_\ell$ implies $(j',t') = (j^{>}_\ell,t^{>}_\ell)$ and $\iota(\ell,j',t') = \iota^*_\ell + 1$ implies $(j',t') = (j^{<}_\ell,t^{<}_\ell)$ (see \cref{obs:rank3_indices}).
On a machine $\sortmach(\ell,1,j',t')$ of its second block $\sblock_{3\ell + 1}$, on the other hand, $\sortjob(\ell,0,j^{>}_\ell,t^{>}_\ell)$ has size 
$2\cdot N^{2\iota(\ell+1,j',t')-2(\iota^*_\ell + 1)} + \eps\cdot N^{2(\iota^*_\ell + 1) - 2\iota(\ell+1,j',t')} + 2\cdot N^{- C + 2\iota^*_\ell -2\iota(\ell+1,j',t')}$, i.e., at least $\eps N^2$ if $\iota(\ell+1,j',t') \neq \iota^*_\ell + 1$ and $2 + \eps + 2 \cdot N^{- C}$ otherwise.
Since $\iota(\ell+1,j',t') = \iota^*_\ell + 1$ implies $(j',t') = (j^{>}_\ell,t^{>}_\ell)$, this is again the correct behavior.
Finally, we consider the sizes of $\sortjob(\ell,0,j^{>}_\ell,t^{>}_\ell)$ on the amplifier machines $\ampmach(\ell,0)$ and $\ampmach(\ell,1)$ of the two blocks, which are 
$2\cdot N^{-3} + \eps\cdot N^{-C + 1} + 2\cdot N$ and
$2\cdot N^{-3} + \eps\cdot N^{3} + 2\cdot N^{-C + 1}$, respectively.

The second special cases we consider deals with $\sortjob(\ell,0,j^{<}_\ell,t^{<}_\ell)$.
On $\sortmach(\ell,0,j',t')$ it has size 
$N^{2\iota(\ell,j',t')-2(\iota^*_\ell + 1)} + N^{-C + 2\iota^*_\ell-2\iota(\ell,j',t')} + \eps\cdot N^{2(\iota^*_\ell + 1)-2\iota(\ell,j',t')}$, i.e., at least $\eps N^2$ if $\iota(\ell,j',t') \neq \iota^*_\ell + 1 = \iota(\ell,j^{<}_\ell,t^{<}_\ell)$ and $1 + N^{-C - 2} + \eps$ otherwise.
On a second block machine $\sortmach(\ell,1,j',t')$, on the other hand, it has size
$N^{2\iota(\ell+1,j',t')-2(\iota^*_\ell + 1)} + N^{2\iota^*_\ell-2 \iota(\ell+1,j',t')} + \eps\cdot N^{-C + 2(\iota^*_\ell + 1)-2\iota(\ell+1,j',t')}$.
Hence, its size is at least $N^2$ if $\iota(\ell+1,j',t') \notin \set{\iota^*_\ell, \iota^*_\ell + 1}$.
If $\iota(\ell+1,j',t') = \iota^*_\ell$, we have a size of $N^{-2} + 1 + \eps\cdot N^{-C + 2}$,
and if $\iota(\ell + 1,j',t') = \iota^*_\ell + 1$, we have a size of $1 + N^{-2} + \eps\cdot N^{-C}$.
Again, note that $\iota^*_\ell = \iota(\ell + 1,j^{<}_\ell,t^{<}_\ell)$ and $\iota^*_\ell + 1 = \iota(\ell + 1,j^{>}_\ell,t^{>}_\ell)$.
The last step is again to consider the sizes of $\sortjob(\ell,0,j^{<}_\ell,t^{<}_\ell)$ on the amplifier machines $\ampmach(\ell,0)$ and $\ampmach(\ell,1)$ of the two blocks, which are 
$ N^{-3} + N^{- C + 1} + \eps\cdot N^3$ and
$ N^{-3} + N + \eps\cdot N^{- C + 3} $, respectively.

Finally, we consider the job $\sortjob(\ell,1,j^{>}_\ell,t^{>}_\ell)$.
However, it is easy to see that we can essentially use the same argumentation as in the very first case considered in this paragraph.

\proofsubparagraph*{Amplifier bridge jobs.}

The blocks of the amplifier bridge job $\ampbridgejob(\ell,q)$ are $\sblock_{3\ell + q}$ and $\sblock_{3\ell + q +1}$.
On a sorting machine $\sortmach(\ell,q,j',t')$ from the first block its size is equal to 
$\eps \cdot N^{2\iota(\ell + \ceil{\frac{q}{2}},j',t') -2\iota^*_\ell +1} + N^{- C + 2\iota^*_\ell - 1 - 2\iota(\ell + \ceil{\frac{q}{2}},j',t') } +N^{ 2\iota^*_\ell - 1 - 2\iota(\ell + \ceil{\frac{q}{2}},j',t') }> \eps N$ (due to first or last addend).
On an amplifier machine $\ampmach(\ell,q)$ from the same block it has a size of $\eps + N^{- C } + 1$.
Moreover, its size on a sorting machine $\sortmach(\ell,q + 1,j',t')$ from the second block is
$\eps \cdot N^{2\iota(\ell+1,j',t')-2\iota^*_\ell +1} + N^{2\iota^*_\ell - 1 - 2\iota(\ell+1,j',t')} +N^{-C + 2\iota^*_\ell - 1 - 2\iota(\ell+1,j',t')} > \eps N$ (due to first or second addend),
and $\eps  + 1 +N^{- C}$ on an amplifier machine $\ampmach(\ell,q + 1)$ from this block.

\proofsubparagraph*{Amplifier shift jobs.}

We consider the jobs $\ampshiftjob(\ell,0)$, $\ampshiftjob(\ell,1)$, and $\ampshiftjob(\ell,2)$ one after another.
The first $\ampshiftjob(\ell,0)$ is associated with $\sblock_{3\ell}$ and on a sorting machine $\sortmach(\ell,0,j',t')$ from this block it has size
$N^{2\iota(\ell,j',t')-2\iota^*_\ell} + \eps \cdot N^{2\iota^*_\ell - 1 - 2\iota(\ell,j',t')} +N^{2\iota^*_\ell - 1 - 2\iota(\ell,j',t')}$, i.e., at least size $N$ if $\iota(\ell,j',t') \neq \iota^*_\ell = \iota(\ell,j^{>}_\ell,t^{>}_\ell)$ and 
$1 + \eps \cdot N^{- 1} +N^{- 1}$ otherwise.
Furthermore, it has size $N^{-1} + \eps  + 1$ on the amplifier machine $\ampmach(\ell,0)$ from this block.

For the second job $\ampshiftjob(\ell,1)$, we have to consider machines from $\sblock_{3\ell + 2}$.
Its size on a sorting machine $\sortmach(\ell,2,j',t')$ from this block is
$N^{2\iota(\ell+1,j',t')-2\iota^*_\ell} + \eps \cdot N^{2\iota^*_\ell - 1 -2\iota(\ell+1,j',t')} + N^{2\iota^*_\ell - 1 -2\iota(\ell+1,j',t')}$, i.e., at least $N$ if $\iota(\ell+1,j',t') \neq \iota^*_\ell = \iota(\ell+1,j^{<}_\ell,t^{<}_\ell)$ and $1 + \eps \cdot N^{- 1 } + N^{- 1}$ otherwise.
Moreover, on the amplifier machine $\ampmach(\ell,2)$ of this block it has size $N^{-1} + \eps + 1$.

Lastly, the job $\ampshiftjob(\ell,2)$ is associated with $\sblock_{3\ell + 2}$ as well.
Its size on a sorting machine $\sortmach(\ell,2,j',t')$ from this block is 
$N^{2\iota(\ell+1,j',t') -2(\iota^*_\ell + 1)} + \eps \cdot N^{2\iota^*_\ell - 2\iota(\ell+1,j',t')} +N^{2\iota^*_\ell - 2\iota(\ell+1,j',t')}$. 
Hence its size is at least $N^2$ if $\iota(\ell+1,j',t')\notin \set{\iota^*_\ell, \iota^*_\ell + 1}$.
If $\iota(\ell+1,j',t') = \iota^*_\ell = \iota(\ell+1,j^{<}_\ell,t^{<}_\ell)$, the size is $N^{-2} + \eps + 1$, and if $\iota(\ell+1,j',t') = \iota^*_\ell + 1 = \iota(\ell+1,j^{>}_\ell,t^{>}_\ell)$, the size is $1 + \eps \cdot N^{-2} +N^{-2}$.
Lastly, we consider the the amplifier machine $\ampmach(\ell,2)$ of this block.
On this machine $\ampshiftjob(\ell,2)$ has a size of $N^{-3} + \eps \cdot N +N > N$.

\proofsubparagraph*{Private loads on sorting machines.}

We consider the private load of machine $\sortmach(\ell,q,j,t)$ from block $\sblock_{3\ell + q}$.
On a sorting machine $\sortmach(\ell,q,j',t')$ from the same block it has a size of 
$N^{2\iota(\ell+ \ceil{\frac{q}{2}},j',t') - 2\iota(\ell+ \ceil{\frac{q}{2}},j,t)} + \eps \cdot N^{2\iota(\ell+ \ceil{\frac{q}{2}},j,t) - 2\iota(\ell+ \ceil{\frac{q}{2}},j',t')} + \eps \cdot N^{2\iota(\ell+ \ceil{\frac{q}{2}},j,t) - 2\iota(\ell+ \ceil{\frac{q}{2}},j',t')}$, i.e., 
at least $\eps N^2$ if $(j,t) \neq (j',t')$ and $1 + 2\eps$ otherwise.
On the amplifier machine $\ampmach(\ell,q)$ of the respective block its size is 
$N^{2\iota^*_\ell  -1 -2\iota(\ell+ \ceil{\frac{q}{2}},j,t)} + \eps \cdot N^{2\iota(\ell+ \ceil{\frac{q}{2}},j,t) -2\iota^*_\ell +1} + \eps \cdot N^{2\iota(\ell+ \ceil{\frac{q}{2}},j,t) -2\iota^*_\ell +1} > \eps N$.
 
\proofsubparagraph*{Private loads on amplifier machines.}

We consider the private load of a machine $\ampmach(\ell,q)$ from block $\sblock_{3\ell + q}$.
On a sorting machine $\sortmach(\ell,q,j',t')$ from the same block it has a size of 
$N^{2\iota(\ell+ \ceil{\frac{q}{2}},j',t') -2\iota^*_\ell +1} + \eps \cdot N^{2\iota^*_\ell - 1- 2\iota(\ell+ \ceil{\frac{q}{2}},j',t')} + \eps \cdot N^{2\iota^*_\ell - 1- 2\iota(\ell+ \ceil{\frac{q}{2}},j',t')} > \eps N$. 
On the amplifier machine $\ampmach(\ell,q)$ of the block, on the other hand, its size is $1 + 2\eps$.

\proofsubparagraph*{Clause jobs.}

Lastly, we consider clause jobs $\clausejob(i,s)$ associated with block $\cblock$.
Its size on a clause machine $\clausemach(i',s')$ is 
$\eps \cdot N^{2(3i' + s') -2(3i + 2)} + \phi(i,s) \cdot N^{i - i'}$. 
Hence, the size is at least $\eps N$ if $i\neq i'$, and between $\eps \cdot N^{-2} + \phi(i,s)$ and $\eps + \phi(i,s)$ otherwise.
\end{proof}

\section{Conclusion}

We conclude this work with a brief discussion of possible future research directions.
There are some obvious questions that can be pursued directly building upon the presented results, i.e., a better approximation ration for \rai or even stronger inapproximability results for \rai or \rar{2}. 
Of course, an improved approximation ratio for any problem of the family would be interesting to develop. 
We would like to highlight \lrs{2}, in particular, as the in some sense easiest problem in the family without a known polynomial time approximation with ratio better than $2$.
Lastly, only very little is known regarding fixed-parameter tractable algorithms for this family of problems. 
For instance, it is open whether \rar{1} is fixed-parameter tractable with respect to the objective value.

\bibliography{rai-II.bib}

\begin{thebibliography}{10}

\bibitem{DBLP:journals/siamcomp/Annamalai19}
Chidambaram Annamalai.
\newblock Lazy local search meets machine scheduling.
\newblock {\em {SIAM} J. Comput.}, 48(5):1503--1543, 2019.
\newblock \href {https://doi.org/10.1137/17M1139175}
  {\path{doi:10.1137/17M1139175}}.

\bibitem{DBLP:conf/stoc/BansalS06}
Nikhil Bansal and Maxim Sviridenko.
\newblock The santa claus problem.
\newblock In Jon~M. Kleinberg, editor, {\em Proceedings of the 38th Annual
  {ACM} Symposium on Theory of Computing, Seattle, WA, USA, May 21-23, 2006},
  pages 31--40. {ACM}, 2006.
\newblock \href {https://doi.org/10.1145/1132516.1132522}
  {\path{doi:10.1145/1132516.1132522}}.

\bibitem{DBLP:conf/soda/BhaskaraKTW13}
Aditya Bhaskara, Ravishankar Krishnaswamy, Kunal Talwar, and Udi Wieder.
\newblock Minimum makespan scheduling with low rank processing times.
\newblock In Sanjeev Khanna, editor, {\em Proceedings of the Twenty-Fourth
  Annual {ACM-SIAM} Symposium on Discrete Algorithms, {SODA} 2013, New Orleans,
  Louisiana, USA, January 6-8, 2013}, pages 937--947. {SIAM}, 2013.
\newblock \href {https://doi.org/10.1137/1.9781611973105.67}
  {\path{doi:10.1137/1.9781611973105.67}}.

\bibitem{DBLP:conf/soda/ChakrabartyKL15}
Deeparnab Chakrabarty, Sanjeev Khanna, and Shi Li.
\newblock On (1, \emph{{$\varepsilon$}})-restricted assignment makespan
  minimization.
\newblock In Piotr Indyk, editor, {\em Proceedings of the Twenty-Sixth Annual
  {ACM-SIAM} Symposium on Discrete Algorithms, {SODA} 2015, San Diego, CA, USA,
  January 4-6, 2015}, pages 1087--1101. {SIAM}, 2015.
\newblock \href {https://doi.org/10.1137/1.9781611973730.73}
  {\path{doi:10.1137/1.9781611973730.73}}.

\bibitem{DBLP:conf/stacs/0011MYZ17}
Lin Chen, D{\'{a}}niel Marx, Deshi Ye, and Guochuan Zhang.
\newblock Parameterized and approximation results for scheduling with a low
  rank processing time matrix.
\newblock In Heribert Vollmer and Brigitte Vall{\'{e}}e, editors, {\em 34th
  Symposium on Theoretical Aspects of Computer Science, {STACS} 2017, March
  8-11, 2017, Hannover, Germany}, volume~66 of {\em LIPIcs}, pages 22:1--22:14.
  Schloss Dagstuhl - Leibniz-Zentrum f{\"{u}}r Informatik, 2017.
\newblock \href {https://doi.org/10.4230/LIPIcs.STACS.2017.22}
  {\path{doi:10.4230/LIPIcs.STACS.2017.22}}.

\bibitem{DBLP:journals/orl/ChenYZ14}
Lin Chen, Deshi Ye, and Guochuan Zhang.
\newblock An improved lower bound for rank four scheduling.
\newblock {\em Oper. Res. Lett.}, 42(5):348--350, 2014.
\newblock \href {https://doi.org/10.1016/j.orl.2014.06.003}
  {\path{doi:10.1016/j.orl.2014.06.003}}.

\bibitem{DBLP:journals/algorithmica/EbenlendrKS14}
Tom{\'{a}}s Ebenlendr, Marek Krc{\'{a}}l, and Jir{\'{\i}} Sgall.
\newblock Graph balancing: {A} special case of scheduling unrelated parallel
  machines.
\newblock {\em Algorithmica}, 68(1):62--80, 2014.
\newblock \href {https://doi.org/10.1007/s00453-012-9668-9}
  {\path{doi:10.1007/s00453-012-9668-9}}.

\bibitem{EpsteinL11}
Leah Epstein and Asaf Levin.
\newblock Scheduling with processing set restrictions: Ptas results for several
  variants.
\newblock {\em International Journal of Production Economics}, 133(2):586--595,
  2011.
\newblock \href {https://doi.org/10.1016/j.ijpe.2011.04.024}
  {\path{doi:10.1016/j.ijpe.2011.04.024}}.

\bibitem{DBLP:conf/soda/Feige08}
Uriel Feige.
\newblock On allocations that maximize fairness.
\newblock In Shang{-}Hua Teng, editor, {\em Proceedings of the Nineteenth
  Annual {ACM-SIAM} Symposium on Discrete Algorithms, {SODA} 2008, San
  Francisco, California, USA, January 20-22, 2008}, pages 287--293. {SIAM},
  2008.
\newblock URL: \url{http://dl.acm.org/citation.cfm?id=1347082.1347114}.

\bibitem{DBLP:journals/jacm/HochbaumS87}
Dorit~S. Hochbaum and David~B. Shmoys.
\newblock Using dual approximation algorithms for scheduling problems
  theoretical and practical results.
\newblock {\em J. {ACM}}, 34(1):144--162, 1987.
\newblock \href {https://doi.org/10.1145/7531.7535}
  {\path{doi:10.1145/7531.7535}}.

\bibitem{DBLP:journals/tcs/JansenMS20}
Klaus Jansen, Marten Maack, and Roberto Solis{-}Oba.
\newblock Structural parameters for scheduling with assignment restrictions.
\newblock {\em Theor. Comput. Sci.}, 844:154--170, 2020.
\newblock \href {https://doi.org/10.1016/j.tcs.2020.08.015}
  {\path{doi:10.1016/j.tcs.2020.08.015}}.

\bibitem{DBLP:journals/siamcomp/JansenR20}
Klaus Jansen and Lars Rohwedder.
\newblock A quasi-polynomial approximation for the restricted assignment
  problem.
\newblock {\em {SIAM} J. Comput.}, 49(6):1083--1108, 2020.
\newblock \href {https://doi.org/10.1137/19M128257X}
  {\path{doi:10.1137/19M128257X}}.

\bibitem{DBLP:journals/corr/KhodamoradiKRS16}
Kamyar Khodamoradi, Ramesh Krishnamurti, Arash Rafiey, and Georgios Stamoulis.
\newblock {PTAS} for ordered instances of resource allocation problems with
  restrictions on inclusions.
\newblock {\em CoRR}, abs/1610.00082, 2016.
\newblock URL: \url{http://arxiv.org/abs/1610.00082}, \href
  {http://arxiv.org/abs/1610.00082} {\path{arXiv:1610.00082}}.

\bibitem{DBLP:journals/anor/LeeLP13}
Kangbok Lee, Joseph~Y.{-}T. Leung, and Michael~L. Pinedo.
\newblock Makespan minimization in online scheduling with machine eligibility.
\newblock {\em Ann. Oper. Res.}, 204(1):189--222, 2013.
\newblock \href {https://doi.org/10.1007/s10479-012-1271-6}
  {\path{doi:10.1007/s10479-012-1271-6}}.

\bibitem{DBLP:journals/mp/LenstraST90}
Jan~Karel Lenstra, David~B. Shmoys, and {\'{E}}va Tardos.
\newblock Approximation algorithms for scheduling unrelated parallel machines.
\newblock {\em Math. Program.}, 46:259--271, 1990.
\newblock \href {https://doi.org/10.1007/BF01585745}
  {\path{doi:10.1007/BF01585745}}.

\bibitem{LeungL08survey}
Joseph Y-T Leung and Chung-Lun Li.
\newblock Scheduling with processing set restrictions: A survey.
\newblock {\em International Journal of Production Economics}, 116(2):251--262,
  2008.
\newblock \href {https://doi.org/10.1016/j.ijpe.2008.09.003}
  {\path{doi:10.1016/j.ijpe.2008.09.003}}.

\bibitem{LeungL16update}
Joseph Y-T Leung and Chung-Lun Li.
\newblock Scheduling with processing set restrictions: A literature update.
\newblock {\em International Journal of Production Economics}, 175:1--11, 2016.
\newblock \href {https://doi.org/10.1016/j.ijpe.2014.09.038}
  {\path{doi:10.1016/j.ijpe.2014.09.038}}.

\bibitem{DBLP:journals/eor/LiW10a}
Chung{-}Lun Li and Xiuli Wang.
\newblock Scheduling parallel machines with inclusive processing set
  restrictions and job release times.
\newblock {\em Eur. J. Oper. Res.}, 200(3):702--710, 2010.
\newblock \href {https://doi.org/10.1016/j.ejor.2009.02.011}
  {\path{doi:10.1016/j.ejor.2009.02.011}}.

\bibitem{DBLP:journals/corr/abs-1907-03526}
Marten Maack and Klaus Jansen.
\newblock Inapproximability results for scheduling with interval and resource
  restrictions.
\newblock {\em CoRR}, abs/1907.03526, 2019.
\newblock URL: \url{http://arxiv.org/abs/1907.03526}, \href
  {http://arxiv.org/abs/1907.03526} {\path{arXiv:1907.03526}}.

\bibitem{DBLP:conf/stacs/MaackJ20}
Marten Maack and Klaus Jansen.
\newblock Inapproximability results for scheduling with interval and resource
  restrictions.
\newblock In Christophe Paul and Markus Bl{\"{a}}ser, editors, {\em 37th
  International Symposium on Theoretical Aspects of Computer Science, {STACS}
  2020, March 10-13, 2020, Montpellier, France}, volume 154 of {\em LIPIcs},
  pages 5:1--5:18. Schloss Dagstuhl - Leibniz-Zentrum f{\"{u}}r Informatik,
  2020.
\newblock \href {https://doi.org/10.4230/LIPIcs.STACS.2020.5}
  {\path{doi:10.4230/LIPIcs.STACS.2020.5}}.

\bibitem{DBLP:journals/orl/MuratoreSW10}
Gabriella Muratore, Ulrich~M. Schwarz, and Gerhard~J. Woeginger.
\newblock Parallel machine scheduling with nested job assignment restrictions.
\newblock {\em Oper. Res. Lett.}, 38(1):47--50, 2010.
\newblock \href {https://doi.org/10.1016/j.orl.2009.09.010}
  {\path{doi:10.1016/j.orl.2009.09.010}}.

\bibitem{SchuurmanW99}
Petra Schuurman and Gerhard~J Woeginger.
\newblock Polynomial time approximation algorithms for machine scheduling: Ten
  open problems.
\newblock {\em Journal of Scheduling}, 2(5):203--213, 1999.
\newblock \href
  {https://doi.org/10.1002/(SICI)1099-1425(199909/10)2:5<203::AID-JOS26>3.0.CO;2-5}
  {\path{doi:10.1002/(SICI)1099-1425(199909/10)2:5<203::AID-JOS26>3.0.CO;2-5}}.

\bibitem{DBLP:phd/dnb/Schwarz10}
Ulrich~M. Schwarz.
\newblock {\em Approximation algorithms for scheduling and two-dimensional
  packing problems}.
\newblock PhD thesis, University of Kiel, 2010.
\newblock URL:
  \url{http://eldiss.uni-kiel.de/macau/receive/dissertation\_diss\_00005147}.

\bibitem{DBLP:journals/siamcomp/Svensson12}
Ola Svensson.
\newblock Santa claus schedules jobs on unrelated machines.
\newblock {\em {SIAM} J. Comput.}, 41(5):1318--1341, 2012.
\newblock \href {https://doi.org/10.1137/110851201}
  {\path{doi:10.1137/110851201}}.

\bibitem{DBLP:journals/ipl/WangS16}
Chao Wang and Ren{\'{e}} Sitters.
\newblock On some special cases of the restricted assignment problem.
\newblock {\em Inf. Process. Lett.}, 116(11):723--728, 2016.
\newblock \href {https://doi.org/10.1016/j.ipl.2016.06.007}
  {\path{doi:10.1016/j.ipl.2016.06.007}}.

\bibitem{DBLP:books/daglib/0030297}
David~P. Williamson and David~B. Shmoys.
\newblock {\em The Design of Approximation Algorithms}.
\newblock Cambridge University Press, 2011.
\newblock URL:
  \url{http://www.cambridge.org/de/knowledge/isbn/item5759340/?site\_locale=de\_DE}.

\end{thebibliography}

\end{document}